\documentclass[fleqn,12pt,twoside]{article}

\usepackage[headings]{espcrc1}
\readRCS $Id: espcrc1.tex,v 1.2 2004/02/24 11:22:11 spepping Exp $
\usepackage{graphicx}
\usepackage[figuresright]{rotating}

\newtheorem{theorem}{Theorem}

\newtheorem{conjecture}{Conjecture}
\newtheorem{corollary}{Corollary}

\newtheorem{definition}{Definition}

\newtheorem{lemma}{Lemma}

\newtheorem{property}{Property}
\newtheorem{remark}{Remark}

\newenvironment{proof}[1][Proof.]{\begin{trivlist}
\item[\hskip \labelsep {\bfseries #1}]}{\end{trivlist}}

\newenvironment{acknowledgement}[1][Acknowledgement]{\begin{trivlist}
\item[\hskip \labelsep {\bfseries #1}]}{\end{trivlist}}

\newcommand{\AmS}{{\protect\the\textfont2
  A\kern-.1667em\lower.5ex\hbox{M}\kern-.125emS}}

\usepackage{amsmath}
\usepackage{amsfonts}
\title{On disjoint matchings in cubic graphs: maximum $2$- and $3$-edge-colorable subgraphs}

\author{Davit Aslanyan \address[MCSD]{Department of Informatics and Applied Mathematics,\\
Yerevan State University, Yerevan, 0025, Armenia}\thanks{email: david@joomag.com}, 
Vahan V. Mkrtchyan\addressmark[MCSD]%
\address{Institute for Informatics and Automation Problems,\\
National Academy of Sciences of Republic of Armenia, 0014,
Armenia}
\thanks{email: vahanmkrtchyan2002@\{ysu.am, ipia.sci.am,
yahoo.com\}},
        Samvel S. Petrosyan\addressmark[MCSD]\thanks{email: samvelpetrosyan2008@yahoo.com},
                and
        Gagik N. Vardanyan \addressmark[MCSD]\thanks{email: vgagik@gmail.com }}

\runtitle{On disjoint matchings in cubic graphs: maximum $2$- and $3$-edge-colorable subgraphs}
\runauthor{Davit Aslanyan, Vahan Mkrtchyan, Samvel Petrosyan, Gagik
Vardanyan}

\begin{document}

\maketitle

\begin{abstract}
We show that any $2-$factor of a cubic graph can be extended to a
maximum $3-$edge-colorable subgraph. We also show that the sum of sizes of
maximum $2-$ and $3-$edge-colorable subgraphs of a cubic graph is at
least twice of its number of vertices. Finally, for a cubic graph $G$, consider the pairs of edge-disjoint matchings whose union consists of as many edges as possible. Let $H$ be the
largest matching among such pairs. Let $M$ be a maximum matching of $G$. We show that $9/8$ is a tight upper bound for $|M|/|H|$.
\end{abstract}

\section{Introduction}

We consider finite undirected graphs that do not contain loops. Graphs may contain multiple edges. For a graph $G$ and a positive integer $k$ define

\begin{center}
$B_{k}(G)\equiv \{(H_{1},...,H_{k}):H_{1},...,H_{k}$ are pairwise
edge-disjoint matchings of $G\}$,
\end{center}

and let

\begin{center}
$\nu _{k}(G)\equiv \max \{\left\vert H_{1}\right\vert
+...+\left\vert H_{k}\right\vert :(H_{1},...,H_{k})\in B_{k}(G)\}$.
\end{center}

A subgraph $H$ of $G$ is called maximum $k$-edge-colorable, if it is $k$-edge-colorable
and contains exactly $\nu _{k}(G)$ edges.

Define:

\begin{center}
$\alpha _{k}(G)\equiv \max \{\left\vert H_{1}\right\vert
,...,\left\vert H_{k}\right\vert :$ $(H_{1},...,H_{k})\in B_{k}(G)$ and
$\left\vert H_{1}\right\vert +...+\left\vert H_{k}\right\vert =\nu
_{k}(G)\}$.
\end{center}

If $\nu(G) $ denotes the cardinality of the largest matching of $G$,
then it
is clear that $\alpha _{k}(G)\leq \nu(G) $ for all $G$ and $k$. Moreover, $%
\nu _{k}(G)=|E(G)|$ for all $k\geq \chi ^{\prime
}(G)$, where $\chi ^{\prime}(G)$ is the chromatic index of $G$. Also note that $\nu _{1}(G)$ and $\alpha _{1}(G)$ coincide with
$\nu (G)$.

Recall that a matching of $G$ is maximum, if it contains $\nu (G)$ edges, and is maximal if it is not a subset of a larger matching. In contrast with the theory of $2$-matchings, where every graph $G$
admits a maximum $2$-matching that includes a maximum matching
\cite{Lov}, there are graphs that do not have a maximum $2$-edge-colorable subgraph that includes a maximum matching.

The following is the best result that can be stated about the ratio
$\nu(G)
/\alpha _{2}(G)$ for any simple graph $G$ (see \cite{FiveFourth}):%
\begin{equation*}
1\leq \nu(G) /\alpha _{2}(G)\leq 5/4.
\end{equation*}

Very deep characterization of simple graphs $G$ satisfying $\nu(G) /\alpha
_{2}(G)=5/4$ is given in \cite{FivefourthCharacter}.

Also note that by Mkrtchyan's result \cite{MPP0-1},
reformulated as in \cite{HararyPlummer}, if $G$ is a matching
covered tree, then $\alpha
_{2}(G)=\nu(G) $. Note that a graph is said to be matching covered (see \cite%
{Perfect}), if its every edge belongs to a maximum matching (not
necessarily a perfect matching as it is usually defined, see e.g.
\cite{Lov}).

In this paper, we show that any $1$- and $2$-factor of a cubic graph can be extended to a maximum $3$-edge-colorable subgraph. We also show that $\nu_{2}(G)+\nu_{3}(G)\geq 2|V(G)|$ for any cubic graph $G$. Finally, we show that $9/8$ is a tight upper bound for the ratio $\nu(G) /\alpha_{2}(G)$ in the class of cubic graphs $G$.

Terms and concepts that we do not define can be found in \cite{Lov,West}.

\section{The main results}

We begin with a theorem that describes the structure of the edges that do not belong to a maximum $3$-edge-colorable subgraph of a cubic graph. 

\begin{theorem}\label{ComplementMatching}Let $H$ be a maximum $3-$edge-colorable
subgraph of a cubic graph $G$. Then $E(G)\backslash E(H)$ is a
matching.
\end{theorem}

\begin{proof} 

To complete the proof of the theorem, we need to verify the absence
of adjacent edges in $G\backslash E(H)$.

Suppose that $(u_0,u_1),(u_1,u_2)\in E(G)\backslash E(H)$. Let $C(u_0),C(u_1),C(u_2)$ denote the colors of the edges incident
to the vertices $u_0,u_1,u_2$, respectively. We need to consider two
cases:

Case 1: $u_0=u_2$, that is, $(u_0,u_1)$ is a multiple edge. Note
that $|C(u_0)|\leq1,|C(u_1)|\leq1$, thus there is $\alpha \in
\{1,2,3\}$ with $\alpha \notin C(u_0)\cup C(u_1)$. Now, if we color
one of edges connecting $u_0$ and $u_1$ with color $\alpha$, then we
would get a proper $3-$edge-coloring of the subgraph $H\cup
\{(u_0,u_1)\}$, contradicting the maximality of $H$.

Case 2: $u_0\neq u_2$. Note that
$|C(u_0)|\leq2,|C(u_1)|\leq1,|C(u_2)|\leq2$. It is easy to see that the maximality of $H$ implies that
\begin{gather*}
C(u_0)\cup C(u_1)=\{1,2,3\} \textrm{ and }C(u_1)\cup
C(u_2)=\{1,2,3\},
\end{gather*}
thus $|C(u_0)|=2,|C(u_1)|=1,|C(u_2)|=2$ and $C(u_0)=C(u_2)$. Suppose that $C(u_0)=C(u_2)=\{\alpha, \beta\}$
and $C(u_1)=\{\gamma\}$. Consider the maximal $\alpha - \gamma$
alternating paths $P_0,P_1,P_2$, starting from vertices
$u_0,u_1,u_2$, respectively. Note that there is $i\in \{0,2\}$ such
that $u_1\notin V(P_i)$. Now, shift the colors on the path $P_i$ to
obtain a new coloring of the maximum $3-$edge-colorable subgraph
$H$, where the color $\alpha$ is absent in both of vertices $u_i$
and $u_1$. Now, if we color the edge $(u_1,u_i)$ with color
$\alpha$, then we would get a proper $3-$edge-coloring of the
subgraph $H\cup \{(u_1,u_i)\}$, contradicting the maximality of $H$.
The proof of the theorem \ref{ComplementMatching} is completed.
\end{proof}

It is not always possible to extend a $1$-factor (and maximum matchings as well \cite{AlbHaas2}) to a maximum $2$-edge-colorable
subgraph of a cubic graph. Nevertheless, the following is true:

\begin{theorem}\label{1factorExtension}Any $1-$factor of a cubic graph $G$ can be extended
to a maximum $3-$edge-colorable subgraph of $G$.
\end{theorem}
\begin{proof}For a $1-$factor $F$ of $G$, choose a maximum $3-$edge-colorable
subgraph $H$ of $G$ with $|E(F)\cap E(H)|$ is maximum.

Let us show that $E(F)\subseteq E(H)$. On the opposite assumption,
consider an edge $e=(u,v)\in E(F)\backslash E(H)$ and assume
that $H$ is properly colored with colors $\{1,2,3\}$. Due to theorem
\ref{ComplementMatching}, the edges adjacent to $e$ belong to
$H$. Let $C(u)$ and $C(v)$ denote the colors of edges that are
incident to $u$ and $v$, respectively. Note that the maximality of
$H$ implies that
\begin{equation*}
|C(u)\cap C(v)|=1\textrm{ and }C(u)\cup C(v)=\{1,2,3\}.
\end{equation*}
Choose $\alpha \in C(u)\backslash C(v)$. Consider the subgraph
$H'=(H\backslash \{e'\})\cup \{e\}$, where $e'$ is the edge that is
incident to $u$ and is colored by $\alpha$. Note that $H'$ is a
maximum $3-$edge-colorable subgraph of $G$ with
\begin{equation*}
|E(F)\cap E(H)|<|E(F)\cap E(H')|.
\end{equation*}contradicting the choice of $H$. The proof of the
theorem \ref{1factorExtension} is completed.
\end{proof}

Next, we prove a result which claims that the uncolored edges with
respect to a maximum $3-$edge-colorable subgraph of $G$ always can
be "left" in a given $1-$factor, or, equivalently, any $2-$factor of
a cubic graph $G$ can also be extended to a maximum $3-$edge-colorable
subgraph of $G$.

\begin{theorem}\label{2factorExtension}Let $F$ be any $1-$factor of a cubic graph $G$, and let $\bar{F}$
be the complementary $2-$factor of $F$. Then there is a maximum
$3-$edge-colorable subgraph $H$ of $G$, such that:
\begin{enumerate}
\item[(a)] $E(H)\cup E(F)=E(G);$

\item[(b)] $E(\bar{F})\subseteq E(H).$

\end{enumerate}
\end{theorem}
\begin{proof} Note that (b) follows from (a), thus we will only prove (a).

For a given $1-$factor $F$ of a cubic graph $G$, consider a maximum
$3-$edge-colorable subgraph $H$ of $G$ such that $|E(F)\cap E(H)|$ is minimum.

To show that $E(H)\cup E(F)=E(G)$, we only need to verify that
$E(F)\cup E(H)\supseteq E(G)$. Assume that there is $e=(u,v)\in E(G)$ such
that $e$ belongs to none of $F$ and $H$, and assume that
$H$ is properly colored with colors $\{1,2,3\}$.

Consider the edges adjacent to $e$. Theorem
\ref{ComplementMatching} implies that these edges belong to $H$. If $C(u)$ and
$C(v)$ denote the colors of edges that are incident to $u$ and $v$,
respectively, then the maximality of $H$ implies that
\begin{equation*}
|C(u)\cap C(v)|=1\textrm{ and }C(u)\cup C(v)=\{1,2,3\}.
\end{equation*}
Suppose that
\begin{equation*}
\{\alpha\}=C(u)\cap C(v), C(u)=\{\alpha,\beta \} \textrm{ and }
C(v)=\{\alpha,\gamma\}.
\end{equation*}

Consider the alternating path $P_e$ with edges of colors $\{\beta, \gamma\}$ that starts from the vertex $u$. Note that the maximality of $H$ implies that $P_e$ must terminate at vertex $v$. Thus $P_e$ is an even path, which together with the edge $e$, forms an odd cycle $C_e$. Let us show that the edges that
are incident to a vertex of $C_e$ and do not lie on $C_e$ must be colored, and therefore they are colored by $\alpha$. If some edge of this kind had no color, then we could have shifted the colors on the cycle $C_e$ and get a maximum $3$-edge-colorable subgraph with two adjacent uncolored edges, contradicting theorem \ref{ComplementMatching}.

To complete the proof of the theorem, we need to consider two cases:

Case 1:$E(C_e)\cap E(F)\neq \emptyset$.

Let $f\in E(C_e)\cap E(F)$. Consider a proper partial
$3-$edge-coloring of the graph $G$ obtained from the coloring of $H$
as follows: $f$ is left uncolored, the edges of the even path
$C_e-f$ are colored $\beta$ and $\gamma$, alternatively, the colors
of the rest of edges are left unchanged. Note that the new partial
$3-$edge-coloring corresponds to a maximum $3-$edge-colorable
subgraph $H'$ of $G$ with
\begin{equation*}
|E(F)\cap E(H')|<|E(F)\cap E(H)|
\end{equation*}contradicting the choice of $H$.

Case 2:$E(C_e)\cap E(F)=\emptyset$.

Note that in this case,
\begin{enumerate}
\item[(I)] the edges that are incident to a
vertex of $C_e$, do not lie on $C_e$ and are colored by $\alpha$,
belong to $F$, which and $E(C_e)\cap E(F)=\emptyset$ imply that:

\item[(II)]all maximum $3-$edge-colorable
subgraphs $H'$ of $G$, which can be obtained from the coloring of
$H$, by leaving any edge $g\in E(C_e)$ uncolored, by coloring the
edges of the even path $C_e-g$ $\beta$ and $\gamma$, alternatively,
and leaving the colors of the rest of edges unchanged, satisfy the
condition $|E(F)\cap E(H')|=|E(F)\cap E(H)|$ is minimum.
\end{enumerate}

Now, we consider a proper partial $3-$edge-coloring $\theta$ of the
graph $G$ obtained from the coloring of $H$ by deleting the colors
of the all edges lying on $C_e$. Since $C_e$ is an odd cycle, there is an $\alpha-\gamma$ alternating path $P_w$ in the
$3-$edge-coloring $\theta$ that starts from a vertex $w\in V(C_e)$
and does not terminate on $C_e$. Choose an edge $x=(w,z)\in E(C_e)$, and
let $y$ be the other edge of $C_e$ that is incident to $w$.

Consider a proper partial $3-$edge-coloring of $G$ obtained from
$\theta$ as follows:
\begin{itemize}
\item shift the colors on the path $P_w$, and clear the color of the edge of
$F$ that is incident to $z$;

\item color $x$ by $\alpha$, and color the edges of the even path
$C_e-x$ by $\beta-\gamma$ alternatively, starting from the edge $y$.
\end{itemize}

It is not hard to see that the new partial $3-$edge-coloring of $G$
corresponds to a maximum $3-$edge-colorable subgraph $H'$ of $G$,
which by (II) satisfies
\begin{equation*}
|E(F)\cap E(H')|<|E(F)\cap E(H)|
\end{equation*}contradicting the choice of $H$. The proof of the
theorem \ref{2factorExtension} is completed.
\end{proof}

We suspect that the theorem \ref{2factorExtension} can be generalized as follows:

\begin{conjecture}Let $F$ be any maximal (not necessarily, maximum) matching of a cubic graph $G$. Then there is a maximum
$3-$edge-colorable subgraph $H$ of $G$, such that $E(H)\cup F=E(G)$.
\end{conjecture}
 
 The problem of estimating the size of maximum $2-$ and $3$-edge-colorable subgraphs of cubic and subcubic graphs has been investigated in \cite{AlbHaas1,AlbHaas2,Measure}. Recently, Rizzi has considered the maximum $3$-edge-colorable subgraph problem in the class of triangle-free subcubic graphs and has got the following results:
 
\begin{theorem}\cite{Rizzi} Let $G$ be a triangle-free graph with $\Delta(G)\leq 3$. Then $\nu_3(G)\geq (1-\frac{2}{3\gamma_0(G)})|E(G)|$, where $\gamma_0(G)$ denotes the odd girth of $G$.
\end{theorem}

\begin{corollary}\cite{Rizzi} Let $G$ be a triangle-free graph with $\Delta(G)\leq 3$. Then $\nu_3(G)\geq \frac{13}{15}|E(G)|$.
\end{corollary}

For subcubic graphs containing no parallel edges, Rizzi has shown:

\begin{theorem}\cite{Rizzi} Let $G$ be a graph with $\Delta(G)\leq 3$ and no parallel edges. Then $\nu_3(G)\geq \frac{6}{7}|E(G)|$.
\end{theorem}

 We have recently considered the maximum $2$- and $3$-edge-colorable subgraph problems in the class of cubic graphs, and got the following result:

\begin{theorem}\label{MainTheoremCubics}\cite{Part1}: For every cubic graph $G$:
\begin{equation*}
\nu _{2}(G)\geq \frac{4}{5}|V(G)|=\frac{8}{15}|E(G)| ,\nu _{3}(G)\geq \frac{7}{6}|V(G)|=\frac{7}{9}|E(G)|.
\end{equation*}
\end{theorem}

There are graphs attaining bounds of the theorem
\ref{MainTheoremCubics}. The graph from figure \ref{Examples
attaining the bounds}a attains the first bound and the graph from
figure \ref{Examples attaining the bounds}b the second bound.

\begin{center}
\begin{figure}[h]
\begin{center}
\includegraphics[height=10pc]{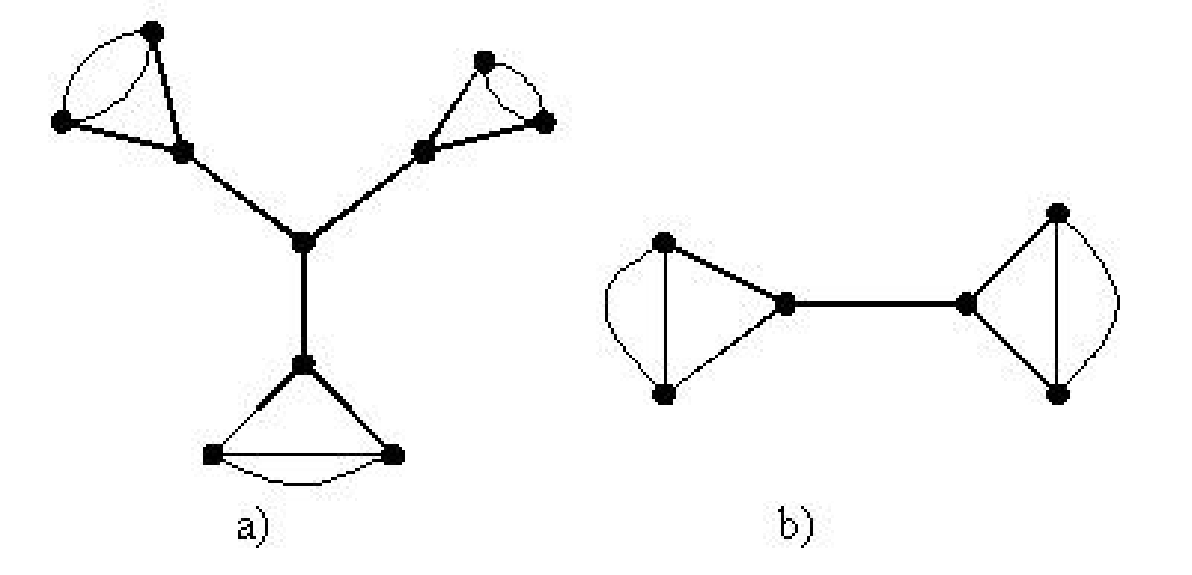}\\
\caption{Examples attaining the bounds of the theorem
\protect\ref{MainTheoremCubics}}\label{Examples attaining the
bounds}
\end{center}
\end{figure}
\end{center}

Note that if there were a cubic graph $G$ attaining the two bounds
at the same time, then we would have:
\begin{equation*}
\nu _{2}(G)+\nu _{3}(G)=\frac{4}{5}|V(G)|+\frac{7}{6}|V(G)|=\frac{59}{30}|V(G)|<2|V(G)|.
\end{equation*}
Thus, to show the absence of a cubic graph attaining all the bounds
of theorem \ref{MainTheoremCubics} at the same time, it suffices to
show the following

\begin{theorem}\label{nu2nu3}
For every cubic graph $G$%
\begin{equation*}
\nu _{2}(G)+\nu _{3}(G)\geq 2|V(G)|.
\end{equation*}
\end{theorem}

\begin{proof}Note that it suffices to prove the statement for connected cubic graphs $G$. If $G$ is $2$-connected, then by Petersen theorem \cite{Lov}, it has a $2$-factor $F$. If $F$ is a Hamilton cycle, it is a 
$2$-edge-colorable subgraph with $|V(G)|$ edges, and the result follows. 

If $F$ has
$k\geq 2$ components, then add an edge $(u,v)\in E(G)$ joining vertices in two different
components of $F$ and delete one edge from each component of $F$, so that
$u$ and $v$ have degree two in the resulting graph. This graph is a collection of
disjoint paths with $|V(G)| + 1-k$ edges and is $2$-edge-colorable. Now consider a
minimal connected subgraph $H$ of $G$ that contains $F$. It is easy to see that
$H$ is $3$-edge-colorable and has $|V(G)|+k-1$ edges. The result follows.

If $G$ is not $2$-connected, it has a cutvertex $v$. Let $T$ be a depth-first-search
tree having root $v$. If there are $k$ vertices of degree three in $T$, then by deleting one
edge incident with each of them, we get a $2$-edge-colorable subgraph with
$|V(G)|-1-k$ edges. The number of leaf nodes in the tree $T$ is $k + 2$ and since
the root is a cutvertex, it has degree at least two in the tree. For each
leaf vertex, add one edge joining it to one of its ancestors in the tree. The
resulting graph is $3$-edge-colorable, as can be seen by inductively coloring
the subgraphs induced by the subtrees of $T$. Since this graph has $|V(G)|-1+k+2$
edges, the result follows. The proof of the theorem \ref{nu2nu3} is completed.
\end{proof}

\begin{remark}The graphs from figure \ref{Examples attaining the
bounds} attain the bounds of the theorem, and we suspect that they
are the only connected cubic graphs with this property.
\end{remark}

The rest of the paper is devoted to obtaining a bound for the ratio $\nu(G) /\alpha _{2}(G)$ in the class of cubic graphs. For this purpose we introduce some definitions and a series of results which will be used for proving theorem \ref{MainBound}.

If $u$ is a vertex of a graph $G$, then let $N(u)$ denote the set of vertices of $G$ that are adjacent to $u$. For a path $P$ let $V_0(P)$ denote the set of end-vertices of $P$. Now, if $V_0(P)=\{u,v\}$ and $w\in V(P)$, then let $P_{u,w}$ denote the subpath of $P$ connecting the vertices $u$ and $w$.
Define:

\begin{center}$M_{2}(G)\equiv \lbrace (H,H^{'}) :(H,H^{'})\in B_{2}(G), 
\vert H\vert +\vert H^{'}\vert =\nu _{2}(G), \vert H\vert =\alpha 
_{2}(G)\rbrace .$\\
\end{center}

Let $A$ and $B$ be matchings of a graph $G$.

\begin{definition}
\label{altpaths} A path (cycle) $e_1, e_2, ..., e_l\ (l\geq1)$ is
called $A$-$B$ alternating if the edges with odd indices belong to $A
\backslash B$ and others to $B \backslash A$, or vice-versa.
\end{definition}

\begin{definition}
\label{maxaltpaths} An $A$-$B$ alternating path $P$ is called maximal
if there is no other $A$-$B$ alternating path that contains $P$ as a proper
subpath.
\end{definition}

The sets of $A$-$B$ alternating cycles and maximal alternating paths are
denoted by $C(A,B)$ and $P(A,B)$, respectively.

The set of the paths from $P(A,B)$ that have even (odd) length is denoted by 
$P_e(A,B)$ ($P_o(A,B)$).

The set of the paths from $P_o(A,B)$ starting from an edge of $A$ (resp. $B$) is
denoted by $P_o^A(A,B)$ (resp. $P_o^B(A,B)$).

Note that every edge $e\in A\cup B$ either belongs to $A\cap B$ or lies on a
cycle from $C(A,B)$ or lies on a path from $P(A,B)$.

Moreover, the following properties are easy to prove \renewcommand{\labelenumi}{(\alph{enumi})}

\begin{property}
\label{AB} \ \renewcommand{\labelenumi}{(\alph{enumi})}

\begin{enumerate}
\item if $F\in C(A,B)\cup P_{e}(A,B)$ then $A$ and $B$ have the same number
of edges that belong to $F$,

\item if $P\in P_{o}^{A}(A,B)$ then the difference between the numbers of
edges that lie on $P$ and belong to $A$ and $B$ is one.
\end{enumerate}
\end{property}

This property immediately implies:

\begin{property}
\label{cardinalitydiff} If $A$ and $B$ are matchings of a graph $G$ then 
\begin{equation*}
|A|-|B|=|P_{o}^{A}(A,B)|-|P_{o}^{B}(A,B)|.
\end{equation*}
\end{property}

Berge's well-known theorem states that a matching $M$ of a graph $G$ is
maximum if and only if $G$ does not contain an $M$-augmenting path \cite{Lov,West}.
This theorem immediately implies:

\begin{property}
\label{maxmatchingproperty} If $M$ is a maximum matching and $H$ is a
matching of a graph $G$ then 
\begin{equation*}
P_{o}^{H}(M,H)=\emptyset,
\end{equation*}%
and therefore, $|M|-|H|=|P_{o}^{M}(M,H)|$.
\end{property}

The proof of the following property is similar to the one of property \ref%
{maxmatchingproperty}:

\begin{property}
\label{HH'} If $(H,H^{\prime})\in M_{2}(G)$ then $P_{o}^{H^{\prime}}(H,H^{%
\prime})=\emptyset $.
\end{property}

Let $G$ be a graph. Over all $(H,H^{\prime })\in M_{2}(G)$ and all maximum
matchings $M$ of $G$, consider the pairs $((H,H^{\prime }),M)$ for which $|M
\cap (H\cup H')|$ is maximized. Among these, choose a pair $((H,H^{\prime }),M)$ such
that $|M \cap H|$ is maximized.

From now on $H,H^{\prime}$ and $M$ are assumed to be chosen as described
above. For this choice of $H,H^{\prime}$ and $M$, consider the paths from $%
P_o^M(M,H)$ and define $M_A$ and $H_A$ as the sets of edges lying on these
paths that belong to $M$ and $H$, respectively.

\begin{lemma}
\label{MHaltpaths} $C(M,H)=P_{e}(M,H)=P_{o}^{H}(M,H)=\emptyset $.
\end{lemma}

\begin{proof}
Property \ref{maxmatchingproperty} implies $P_{o}^{H}(M,H)=\emptyset$. Let
us show that $C(M,H)=P_{e}(M,H)=\emptyset $. Suppose that there is $F_{0}\in
C(M,H)\cup P_{e}(M,H)$. Define: 
\begin{equation*}
M^{\prime }\equiv \lbrack M\backslash E(F_{0})]\cup \lbrack H\cap E(F_{0})].
\end{equation*}

Consider the pair $((H,H^{\prime }),M^{\prime })$. Note that $M^{\prime}$ is
a maximum matching, and 
\begin{equation*}
|(H\cup H') \cap M^{\prime}| \geq |(H\cup H') \cap M|,
\end{equation*}
thus taking into account the choice of the pair $((H,H^{\prime }),M)$, we must have equality. However, this is a contradiction since for this new pair $((H,H^{\prime }),M^{\prime })$, we have that $|H \cap M'|> |H \cap M|$ contradicting $|H \cap M|$ being maximum.
\end{proof}

\begin{corollary}
\label{outerHMs} $M \cap H = M \backslash M_A = H \backslash H_A$.
\end{corollary}

\begin{lemma}
\label{M_AandH's} Each edge of $\ M_{A}\backslash H^{\prime }$ is adjacent
to two edges of $H^{\prime }$.
\end{lemma}

\begin{proof}
Let $e$ be an arbitrary edge from $M_{A}\backslash H^{\prime }$. Note that $%
e\in M$, $e\notin H$, $e\notin H^{\prime }$. Now, if $e$ is not adjacent to
an edge of $H^{\prime }$, then $H\cap (H^{\prime }\cup \{e\})=\emptyset $
and 
\begin{equation*}
|H|+|H^{\prime }\cup \{e\}|>|H|+|H^{\prime }|=\nu _{2}(G),
\end{equation*}%
which contradicts $(H,H^{\prime}) \in M_2(G)$.

On the other hand, if $e$ is adjacent to only one edge $f\in H^{\prime }$,
then consider the pair $(H,H^{\prime \prime })$, where $H^{\prime \prime
}\equiv (H^{\prime }\backslash \{f\}) \cup \{e\}$. Note that 
\begin{equation*}
H \cap H^{\prime \prime }=\emptyset, \ |H^{\prime\prime}| = |H^{\prime}|
\end{equation*}
and 
\begin{equation*}
|(H\cup H^{\prime\prime})\cap M| > |(H\cup H^{\prime})\cap M|,
\end{equation*}
which contradicts $|(H\cup H^{\prime})\cap M|$ being maximum.
\end{proof}

\begin{lemma}
\label{onlyodd} $C(M_A, H^{\prime}) = P_e(M_A, H^{\prime}) = P_o^{M_A}(M_A,
H^{\prime}) = \emptyset$.
\end{lemma}

\begin{proof}
First of all, let us show that $C(M_A,H^{\prime}) = P_e(M_A, H^{\prime}) =
\emptyset$. For the sake of contradiction, suppose that there is $F_{0}\in
C(M_A, H^{\prime}) \cup P_e(M_A, H^{\prime})$. Define: 
\begin{equation*}
H^{\prime\prime}\equiv [ H^{\prime}\backslash E(F_0)] \cup [ M_A \cap
E(F_0)].
\end{equation*}

Consider the pair of matchings $(H,H^{\prime \prime })$. Note that by
the definition of an alternating path we have $M_{A}\cap H=\emptyset $,
therefore 
\begin{equation*}
H\cap H^{\prime \prime }=\emptyset,
\end{equation*}
\begin{equation*}
|H| + |H^{\prime\prime}| = |H| + |H^{\prime}| = \nu _{2}(G)
\end{equation*}
(see property \ref{AB}(a)).

Thus $(H,H^{\prime \prime })\in M_{2}(G)$ and%
\begin{equation*}
|(H\cup H^{\prime\prime})\cap M| > |(H\cup H^{\prime})\cap M|,
\end{equation*}%
which contradicts $|(H\cup H^{\prime})\cap M|$ being maximum.

On the other hand, the end-edges of a path from $P_o^{M_A}(M_A, H^{\prime})$
are from $M_A$ and are adjacent to only one edge from $H^{\prime}$
contradicting lemma \ref{M_AandH's}. Therefore, $P_o^{M_A}(M_A,
H^{\prime})=\emptyset$.
\end{proof}

\begin{lemma}
\label{equality} $|H^{\prime }|=|P_{o}^{H^{\prime }}(M_{A},H^{\prime
})|+|H_{A}|+\nu (G)-\alpha _{2}(G)$.
\end{lemma}

\begin{proof}
Property \ref{cardinalitydiff} implies that
\begin{equation*}
|H^{\prime }|-|M_{A}|=|P_{o}^{H^{\prime }}(M_{A},H^{\prime
})|-|P_{o}^{M_{A}}(M_{A},H^{\prime })|,
\end{equation*}%
and due to property \ref{AB}(b), property \ref{maxmatchingproperty} 
\begin{equation*}
|M_{A}|-|H_{A}|=|P_{o}^{M}(M,H)|=|M|-|H|=\nu (G)-\alpha _{2}(G).
\end{equation*}%
By lemma \ref{onlyodd} $P_{o}^{M_{A}}(M_{A},H^{\prime })=\emptyset $,
therefore, 
\begin{equation*}
|H^{\prime }|=|P_{o}^{H^{\prime }}(M_{A},H^{\prime
})|+|M_{A}|=|P_{o}^{H^{\prime }}(M_{A},H^{\prime })|+|H_{A}|+\nu (G)-\alpha
_{2}(G).
\end{equation*}
\end{proof}

\begin{lemma}
\label{MHpaths} Let $P\in P_{o}(M,H)$ and assume that $P=m_1, h_1, m_2, ...,
h_{l-1}, m_l$, $l \geq 1, m_i \in M, 1 \leq i \leq l, h_j \in H, 1 \leq j
\leq l-1$. Then $l\geq 3$ and $\{ m_1, m_l \} \subseteq H^{\prime}$.
\end{lemma}

\begin{proof}
We claim that $l\geq 3$. Note that if $l=1$ then $P=m_{1}$, and clearly $%
m_{1}\in H^{\prime }$ as otherwise we could enlarge $H$ by adding $m_{1}$ to
it which contradicts $(H,H^{\prime })\in M_{2}(G)$. Thus $l\geq 2$. Now, let us show that $m_{1}\in
H^{\prime }$. If $m_{1}\notin H^{\prime }$ then define%
\begin{equation*}
H_{1}\equiv (H\backslash \{h_{1}\})\cup \{m_{1}\}.
\end{equation*}

Note that%
\begin{equation*}
H_{1}\cap H^{\prime }=\emptyset, \ |H_1| = |H|,
\end{equation*}%
and 
\begin{equation*}
|(H_1\cup H') \cap M| > |(H\cup H') \cap M|
\end{equation*}

which contradicts $|(H\cup H') \cap M|$ being maximum. Thus $m_1 \in H^{\prime}$.
Similarly, one can show that $m_l \in H^{\prime}$.

By property \ref{HH'} $P_{o}^{H^{\prime }}(H,H^{\prime })=\emptyset $,
there is an $i,$ $1\leq i\leq l$, such that $m_i \in M \backslash (H \cup
H^{\prime})$. Since $\{m_1, m_l\} \subseteq H^{\prime}$, we have $l \geq 3$.
\end{proof}

\begin{corollary}
\label{HAbound} $\left\vert H_{A}\right\vert \geq 2(\nu (G)-\alpha _{2}(G))$.
\end{corollary}

\begin{proof}
Due to lemma \ref{MHpaths} every path $P\in P_{o}(M,H)$ has length at least
five, therefore it contains at least two edges from $H$. By property \ref%
{maxmatchingproperty}, there are 
\begin{equation*}
\left\vert P_{o}(M,H)\right\vert =\left\vert P_{o}^{M}(M,H)\right\vert =\nu
(G)-\alpha _{2}(G),
\end{equation*}

paths from $P_{o}(M,H)$, therefore 
\begin{equation*}
\left\vert H_{A}\right\vert \geq 2(\nu (G)-\alpha _{2}(G)).
\end{equation*}
\end{proof}

\begin{corollary}
\label{VerticesMH} Every vertex lying on a path from $P(M,H)=P_{o}^{M}(M,H)$
is incident to an edge from $H^{\prime }$.
\end{corollary}

\begin{proof}
Suppose $w$ is a vertex lying on a path from $P(M,H)=P_{o}^{M}(M,H)$ and
assume that $e$ is an edge from $M_{A}$ incident to the vertex $w$. Clearly,
if $e\in H^{\prime }$ then the corollary is proved, thus we may assume
that $e\notin H^{\prime }$. Note that $e\in M_{A}\backslash H^{\prime }$
therefore by lemma \ref{M_AandH's} $e$ is adjacent to two edges from $%
H^{\prime }$. Thus $w$ is incident to an edge from $H^{\prime }$.
\end{proof}

Let $Y=Y(M,H,H')$ denote the set of the paths from $P(H, H^{\prime})$ starting from
the end-edges of the paths from $P_o^M(M, H)$. Note that $Y$ is well-defined
since by lemma \ref{MHpaths} these end-edges belong to $H^{\prime}$.
According to property \ref{HH'}, $Y \subseteq P_e(H, H^{\prime})$, thus the
set of the last edges of the paths from $Y$ is a subset of $H$. 
Denote this subset by $H_Y$.

\begin{lemma}
\label{H'HpathsandPoddMAH'} \ \renewcommand{\labelenumi}{(\alph{enumi})}

\begin{enumerate}
\item \label{H'Hpaths} $|Y| = 2(\nu(G) - \alpha_2(G))$ and the length of the
paths from $Y$ is at least four,

\item \label{PoddMAH'} $\left\vert P_{o}^{H^{\prime }}(M_{A},H^{\prime
})\right\vert \geq \nu (G)-\alpha _{2}(G).$
\end{enumerate}
\end{lemma}

\begin{proof}
(a) Due to property %
\ref{HH'}, all end-edges of the paths from $P_{o}^{M}(M,H)$ lie on different
paths from $Y$. Therefore $|Y|=2|P_{o}^{M}(M,H)|=2(\nu (G)-\alpha _{2}(G))$.

Since the edges from $H_Y$ are adjacent to only one edge from $H^{\prime}$,
we conclude that they do not lie on a path from $P_o^M(M,H)$ (corollary \ref%
{VerticesMH}). Thus, by corollary \ref{outerHMs}, $H_Y \subseteq M\cap H$%
. Furthermore, as the first two edges of a path from $Y$ lie on a path from $%
P_o^M(M,H)$, and the last edge does not, we conclude that its length is at
least four.

(b) From $H_Y
\subseteq M\cap H$ we get 
\begin{equation*}
\left\vert M\cap H\right\vert \geq |H_Y| = |Y| = 2|P_o^M(M, H)| = 2(\nu
(G)-\alpha _{2}(G)).
\end{equation*}

On the other hand, every edge from $H_Y$ is adjacent to an edge from $%
H^{\prime }\backslash M$, which is an end-edge of a path from $%
P_o^{H^{\prime}}(M_A, H^{\prime})$, therefore%
\begin{equation*}
2(\nu (G)-\alpha _{2}(G))\leq \left\vert M\cap H\right\vert \leq 2\left\vert
P_{o}^{H^{\prime }}(M_{A},H^{\prime })\right\vert
\end{equation*}

or%
\begin{equation*}
\nu (G)-\alpha _{2}(G)\leq \left\vert P_{o}^{H^{\prime }}(M_{A},H^{\prime
})\right\vert .
\end{equation*}
\end{proof}

Up to now, we assumed that $G$ is an arbitrary graph. Now, we turn back to the problem of bounding the ratio $\nu(G) /\alpha _{2}(G)$ in the class of cubic graphs, and assume that $G$ is an arbitrary cubic graph. Moreover, we also assume that the pair $((H,H'),M)$ is chosen as above. In addition, we make one more assumption on the choice of the pair: $((H,H'),M)$. We assume that $\sum_{P\in Y(M,H,H^{'})}{\vert P\vert }$ is maximized subject to previous conditions.

\begin{lemma}\label{EvenPathCase} Let $P\in P_{e}(H,H^{'})$ and $v\in V_{0}(P)$. 
Then there is no edge $(v,w)\in E(G)$, such that $w$ lies on a path $P^{'}\in P_{e}(H,H^{'})$, $P^{'}\ne P$, and the length of the path $P'_{v',w}$ is even, where $v'\in V_0(P')$.
\end{lemma}

\begin{proof} Assume that the length of the path $P'_{v^{'},w}$ is even. We can assume that 
the vertex $v^{'}$ is that end-vertex of $P^{'}$, which is incident to an edge from $
H^{'}$. Without loss of generality, we can also assume that the vertex $v$ is incident to an edge from $H^{'}$, since in the other case we can exchange the edges between $H$ and $H'$ on the path $P$. Note that in this case the vertex $w$, which is the other end-vertex of path $P'_{v^{'},w}$, is incident to an edge from the set $E(P_{v^{'},w})\cap H$.

Define:

\begin{center}$H_{1}\equiv (H\backslash E(P_{v^{'},w}))\cup \lbrace 
(v,w)\rbrace \cup (H^{'}\cap E(P_{v^{'},w}))$,\end{center}

\begin{center}$H_{1}^{'}\equiv (H^{'}\backslash E(P_{v^{'},w}))\cup 
(H\cap E(P_{v^{'},w}-w))$.\end{center}

Note that:

\begin{center}$H_{1}\cap H_{1}^{'}=\emptyset ,$\\
\end{center}

\begin{center}$\vert H_{1}\vert +\vert H_{1}^{'}\vert =\vert H\vert 
+\vert H^{'}\vert =\nu _{2}(G)$ \end{center}

but

\begin{center}$\vert H_{1}\vert >\vert H\vert =\alpha _{2}(G)$,
\end{center}

which contradicts the condition $(H,H^{'})\in M_{2}(G)$.
\end{proof}

\begin{lemma}\label{OddPathCase} Let $ P\in P_{o}^H(H,H^{'})$. Then there are no paths $P_{1},P_{2}\in P_{e}(H,H^{'})$ and $v_{1}\in V_{0}(P_{1})$, $v_{2}\in V_{0}(P_{2})$, such that
$(v_{1},w_{1})\in E(G)$, $
(v_{2},w_{2})\in E(G)$ and $(w_{1},w_{2})\in E(P)$.
\end{lemma}

\begin{proof}
Suppose that there are such paths $P_{1}$ and $P_{2}$. Since $P_{1}$ and $P_{2}$ are even paths, we can assume, that on the path $P_{1}$ the vertex $v_{1}$ is incident to an edge 
from $H$, and on the path $P_{2}$ the vertex $v_{2}$ is incident to an 
edge from $H^{'}$. Note that this assumption is also true, when $P_{1}=P_{2}$. 

Assume, that $V_{0}(P)=\lbrace u_{1},u_{2}\rbrace $ and the
path $P$ connects vertices $u_{1}$ and $u_{2}$, passing through $
w_{1}$ then $w_{2}$ (figure \ref{OddPathCaseFigure}).

\begin{center}
\begin{figure}[h]
\centering
\includegraphics[width=16.52cm,height=6.26cm]{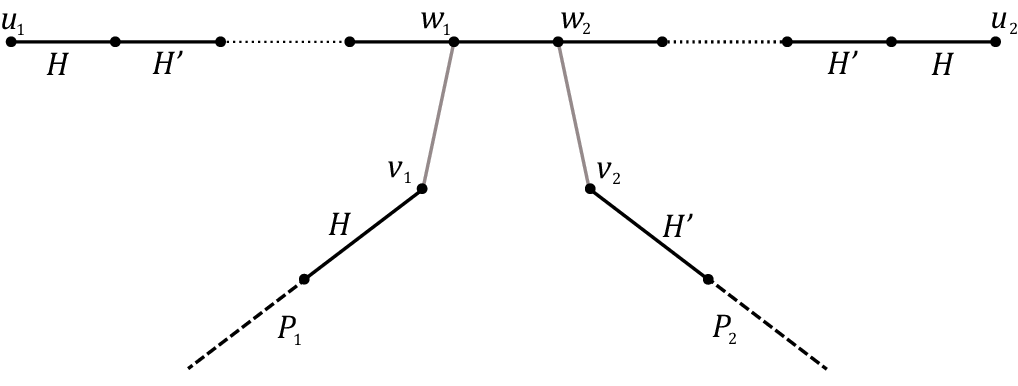}
\caption{The typical situation in the proof of lemma \ref{OddPathCase}}
\label{OddPathCaseFigure}
\end{figure}
\end{center}

We will consider two cases.

Case A: If $(w_{1},w_{2})\in H$.

Define:

\begin{center}$H_{1}\equiv (H\backslash(E(P_{u_{1}w_{1}})\cup \lbrace 
(w_{1},w_{2})\rbrace ))\cup \lbrace (v_{2},w_{2})\rbrace \cup (H^{'}\cap 
E(P_{u_{1}w_{1}})),$\\
\end{center}

\begin{center}$H_{1}^{'}\equiv (H^{'}\backslash E(P_{u_{1}w_{1}}))\cup 
\lbrace (v_{1},w_{1})\rbrace \cup (H\cap E(P_{u_{1}w_{1}}))$.
\end{center}

Note that:

\begin{center}$H_{1}\cap H_{1}^{'}=\emptyset $,\end{center}

but

\begin{center}$\vert H_{1}\vert +\vert H_{1}^{'}\vert >\vert H\vert 
+\vert H^{'}\vert =\nu _{2}(G)$, \end{center}

which contradicts the condition $(H,H^{'})\in M_{2}(G)$.

Case B: If $(w_{1},w_{2})\in H'$.

Define:
 
\begin{center}$H_{1}\equiv (H\backslash E(P_{u_{2}w_{2}}))\cup \lbrace 
(v_{2},w_{2})\rbrace \cup (H^{'}\cap E(P_{u_{2}w_{2}})),$\end{center}

\begin{center}$H_{1}^{'}\equiv (H^{'}\backslash(E(P_{u_{2}w_{2}})\cup 
\lbrace (w_{1},w_{2})\rbrace ))\cup \lbrace (v_{1},w_{1})\rbrace \cup 
(H\cap E(P_{u_{2}w_{2}}))$.\end{center}

Note that:

\begin{center}$H_{1}\cap H_{1}^{'}=\emptyset $,\end{center}

but

\begin{center}$\vert H_{1}\vert +\vert H_{1}^{'}\vert >\vert H\vert 
+\vert H^{'}\vert =\nu _{2}(G)$, \end{center}

which contradicts the condition $(H,H^{'})\in M_{2}(G)$.
\end{proof}

\begin{lemma}\label{CycleCaseDiffPaths} Let $C\in C(H,H^{'})$. Then there are no different paths
$P_{1},P_{2}\in P_{e}(H,H^{'})$ and $v_{1}\in V_{0}(P_{1})$, $v_{2}\in V_{0}(P_{2})$, such that $(v_{1},w_{1})\in E(G)$ and $(v_{2},w_{2})\in E(G)$, where $\lbrace w_{1},w_{2}\rbrace \subseteq V(C)$ and the length of a path $P_{w_{1},w_{2}}$, which connects vertices $w_{1},w_{2}$, all whose edges lie on the cycle $C$, is odd. 
\end{lemma}

\begin{proof}
Assume that there are two such paths $P_{1}$ and $P_{2}$. 

Let us consider two cases.

Case A: $(w_{1},w_{2})\in E(G)$ and $(w_{1},w_{2})\in H$ (figure \ref{CaseACycleCaseDiffPathsFigure}).

\begin{center}
\begin{figure}[h]
\centering
\includegraphics[width=13.15cm,height=6.47cm]{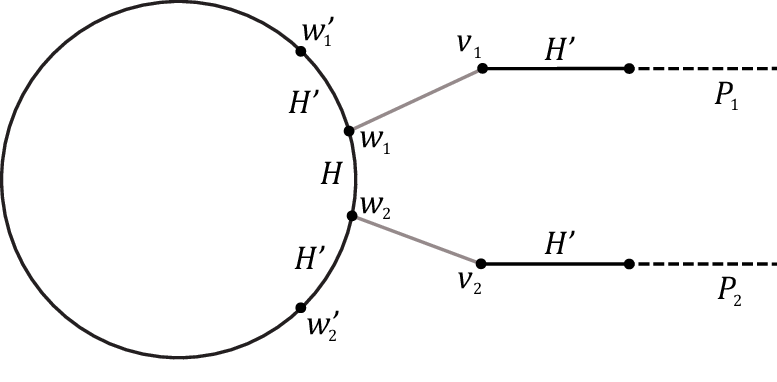}
\caption{Lemma \ref{CycleCaseDiffPaths}, Case A}
\label{CaseACycleCaseDiffPathsFigure}
\end{figure}
\end{center}

Since the paths $P_{1}$ and $P_{2}$ are even, we can assume, 
that on the paths $P_{1}$ and $P_{2}$ both vertices $v_{1}$ and $v_{2}$ are incident to edges from $H'$.

Define:

\begin{center}$H_{1}\equiv (H\backslash\lbrace (w_{1},w_{2})\rbrace 
)\cup \lbrace (v_{1},w_{1}), (v_{2},w_{2})\rbrace $.\end{center}

Note that:

\begin{center}$H_{1}\cap H^{'}=\emptyset $,\end{center}

but

\begin{center}$\vert H_{1}\vert +\vert H^{'}\vert >\vert H\vert +\vert 
H^{'}\vert =\nu _{2}(G),$ \end{center}

which contradicts the condition $(H,H^{'})\in M_{2}(G)$.

Case B: $(w_{1},w_{2})\notin H$ (figure \ref{CaseBCycleCaseDiffPathsFigure}). 

Note that this case contains the subcase $(w_{1},w_{2})\in E(G)$ and $(w_{1},w_{2})\in H^{'}$.

\begin{center}
\begin{figure}[h]
\centering
\includegraphics[width=13.15cm,height=7.82cm]{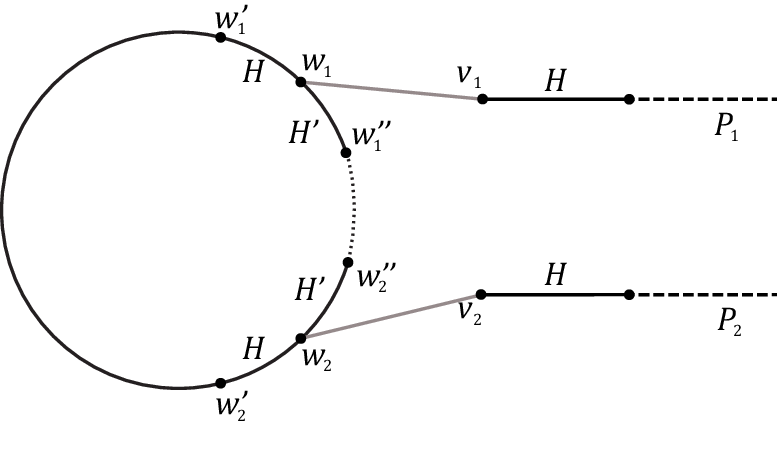}
\caption{Lemma \ref{CycleCaseDiffPaths}, Case B}
\label{CaseBCycleCaseDiffPathsFigure}
\end{figure}
\end{center}

Since the paths $P_{1}$ and $P_{2}$ are even, we can assume, 
that on the paths $P_{1}$ and $P_{2}$ both vertices $v_{1}$ and $v_{2}$ are incident to edges from $H$.

Note that in this case the number of edges of cycle $C$ that belong to 
$H$ is greater or equal to two. Moreover, as $(w_{1},w_{2})\notin H$, then $
(w_{1},w_{1}^{'})\ne (w_{2},w_{2}^{'})$. Note that it is possible that $(w_{1},w_{1}^{''})=(w_{2},w_{2}^{''})$.

Define:

\begin{center}$H_{1}^{'}\equiv (H'\backslash\lbrace 
(w_{1},w_{1}^{''}),(w_{2},w_{2}^{''})\rbrace )\cup \lbrace 
(v_{1},w_{1}), (v_{2},w_{2})\rbrace $.\end{center}

Note that:

\begin{center}$(H,H_{1}^{'})\in M_{2}(G)$,\end{center}

but

\begin{center}$P_{o}^{H_1^{'}}(H,H_{1}^{'})\ne \emptyset $,
\end{center}

which contradicts property \ref{HH'}.
\end{proof}

\begin{lemma}\label{CycleCaseTwoPathsY} Let $C\in C(H,H^{'})$. Then there are no paths $P_{1},P_{2}\in Y=Y(M,H,H^{'})$ and $v_{1}\in V_{0}(P_{1}),$ $v_{2}\in V_{0}(P_{2})$, such that $(v_{1},w_{1})\in E(G)$, $(v_{2},w_{2})\in E(G)$ and $(w_{1},w_{2})\in E(C)$.
\end{lemma}

\begin{proof}
If the paths $P_1$ and $P_2$ are different, then the statement of the lemma follows from lemma \ref{CycleCaseDiffPaths}. 

So we can assume that $P_{1}=P_{2}=P$. We can also assume that on the path $P$ the vertex $v_{1}$ is incident to an edge from $H$, and the vertex $v_{2}$ is incident to an edge from $H^{'}$. Recall that lemma \ref{MHpaths} and the proof of the lemma \ref{H'HpathsandPoddMAH'} imply that both these edges belong to $M$ (figure \ref{CycleCaseTwoPathsYFigure}).

\begin{center}
\begin{figure}[h]
\centering
\includegraphics[width=11.95cm,height=6.47cm]{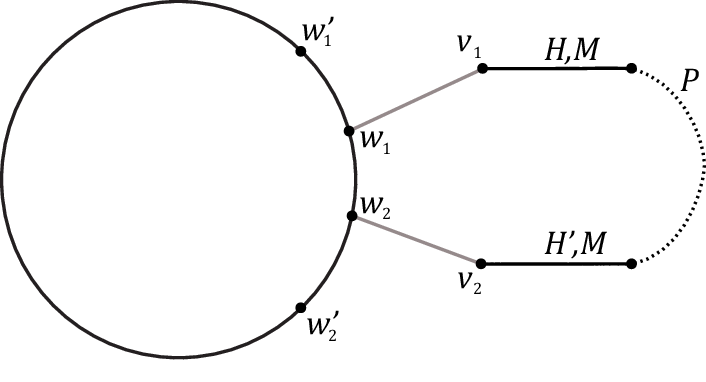}
\caption{The situation in the proof of the lemma \ref{CycleCaseTwoPathsY}}
\label{CycleCaseTwoPathsYFigure}
\end{figure}
\end{center}

We will consider two cases.

Case A: $(w_{1},w_{2})\notin M$.

If $(w_{1},w_{2})\in H'$, then define

\begin{center}$H_{1}\equiv H$,\end{center}

\begin{center}$H_{1}^{'}\equiv (H^{'}\backslash\lbrace 
(w_{1},w_{2})\rbrace )\cup \lbrace (v_{1},w_{1})\rbrace $.\end{center}

Otherwise, when $(w_{1},w_{2})\in H$, define

\begin{center}$H_{1}\equiv (H\backslash\lbrace (w_{1},w_{2})\rbrace 
)\cup \lbrace (v_{2},w_{2})\rbrace $,\end{center}

\begin{center}$H_{1}^{'}\equiv H^{'}$.\end{center}

Note that:

\begin{center}$H_{1}\cap H_{1}^{'}=\emptyset ,$\end{center}

\begin{center}$\vert H_{1}\vert +\vert H_{1}^{'}\vert =\nu _{2}(G), $\end{center}

\begin{center}$\vert H_{1}\vert =\alpha _{2}(G), $\end{center}

\begin{center}$\vert M\cap (H_{1}\cup H_{1}^{'})\vert =\vert M\cap (H\cup H^{'})\vert 
,$\end{center}

\begin{center}$\vert M\cap H_{1}\vert =\vert M\cap H\vert ,$\end{center}

but:

\begin{center}$\sum_{P'\in Y(M,H_{1},H_{1}^{'})}{\vert P'\vert }> 
\sum_{P'\in Y(M,H,H^{'})}{\vert P'\vert }$,\end{center}

which contradicts the choice of the pair $((H,H^{'}),M)$.

Case B: $(w_{1},w_{2})\in M$.

If $(w_{1},w_{2})\in H'$, then define:

\begin{center}$H_{1}\equiv (H\backslash\lbrace 
(w_{2},w_{2}^{'})\rbrace )\cup \lbrace (v_{2},w_{2})\rbrace $,
\end{center}

\begin{center}$H_{1}^{'}\equiv H^{'}$.\end{center}

Otherwise, when $(w_{1},w_{2})\in H$, define

\begin{center}$H_{1}\equiv H$,\end{center}

\begin{center}$H_{1}^{'}\equiv (H^{'}\backslash\lbrace (w_{1},w_{1}^{'})\rbrace )\cup 
\lbrace (v_{1},w_{1})\rbrace $\end{center}

Note, that:

\begin{center}$H_{1}\cap H_{1}^{'}=\emptyset ,$\end{center}

\begin{center}$\vert H_{1}\vert +\vert H_{1}^{'}\vert =\nu _{2}(G), $\end{center}

\begin{center}$\vert H_{1}\vert =\alpha _{2}(G), $\end{center}

\begin{center}$\vert M\cap (H_{1}\cup H_{1}^{'})\vert =\vert M\cap (H\cup H^{'})\vert 
,$\end{center}

\begin{center}$\vert M\cap H_{1}\vert =\vert M\cap H\vert ,$\end{center}

but:

\begin{center}$\sum_{P'\in Y(M,H_{1},H_{1}^{'})}{\vert P'\vert }> 
\sum_{P'\in Y(M,H,H^{'})}{\vert P'\vert }$,\end{center}

which contradicts the choice of the pair $((H,H^{'}),M)$.
\end{proof}

\begin{lemma}\label{MANeighbourhood} Let $e=(u,v)\in M_{A}\backslash H^{'}$ and assume that $N(u)= \lbrace 
u_{1},u_{2},v\rbrace $, $N(v)= \lbrace v_{1},v_{2},u\rbrace $ (figure \ref{eNeighbourhood}). 
Then there is no path $P\in Y(M,H,H^{'})$ and $w\in V_{0}(P)$ such that 
$w$ is adjacent to one of the vertices $u_{1},u_{2}, v_{1},v_{2}$.
\end{lemma}

\begin{center}
\begin{figure}[h]
\centering
\includegraphics[width=10.68cm,height=5.06cm]{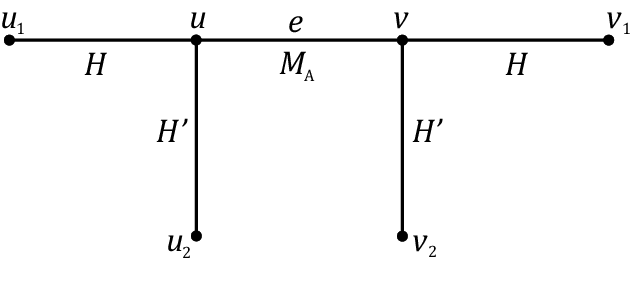}
\caption{The neighbourhood of the edge $e=(u,v)$}
\label{eNeighbourhood}
\end{figure}
\end{center}

\begin{proof} First of all note that for the proof of the lemma, we only
need to consider the case of vertices $u_{1}$ and $u_{2}$. 

Assume the opposite. Then there is a path $P\in Y(M,H,H^{'})$ and $w\in V_{0}(P)$ contradicting the statement.

Consider two cases.

Case A: $(u,u_{1})\notin E(P),(u,u_{2})\notin E(P)$. 

We need to consider two sub-cases.

Sub-case A1: $(w,u_{1})\in E(G)$.

As $P\in P_{e}(H,H^{'})$, we can assume that the vertex $w$ is 
incident to an edge from $H^{'}$.

Define:

\begin{center}$H_{1}\equiv (H\backslash\lbrace (u,u_{1}),(v,v_{1})\rbrace )\cup 
\lbrace e,(w,u_{1})\rbrace .$\end{center}

Note that:

\begin{center}$H_{1}\cap H^{'}=\emptyset ,$\end{center}

\begin{center}$\vert H_{1}\vert +\vert H^{'}\vert =\vert H\vert +\vert H^{'}\vert 
=\nu _{2}(G), $\end{center}

\begin{center}$\vert H_{1}\vert =\vert H\vert =\alpha _{2}(G)$, 
\end{center}

but:

\begin{center}$\vert M\cap (H_{1}\cup H^{'})\vert >\vert M\cap (H\cup H^{'})\vert ,$\end{center}

which contradicts the choice of the pair $((H,H^{'}),M)$.

Sub-case A2: $(w,u_{2})\in E(G)$.

As $P\in P_{e}(H,H^{'})$, we can assume that the vertex $w$ is 
incident to an edge from $H^{'}$.

Define:

\begin{center}$H_{1}^{'}\equiv (H^{'}\backslash\lbrace (u,u_{2}),(v,v_{2})\rbrace 
)\cup \lbrace e,(w,u_{2})\rbrace .$\end{center}

Note, that:

\begin{center}$H\cap H_{1}^{'}=\emptyset , $\end{center}

\begin{center}$\vert H\vert +\vert H_{1}^{'}\vert =\vert H\vert +\vert 
H^{'}\vert =\nu _{2}(G)$,\end{center}

but:

\begin{center}$\vert M\cap (H\cup H_{1}^{'})\vert >\vert M\cap (H\cup H^{'})\vert ,$\end{center}

which contradicts the choice of the pair $((H,H^{'}),M)$.

Case B: $(u,u_{1})\in E(P),(u,u_{2})\in E(P)$. 

We need to consider two sub-cases.

Sub-case B1: The vertex $w$ is incident to an edge from the set $H^{'}\cap M$. 

In this case we have two further sub-cases.

Sub-case B1.1: $(w,u_{1})\in E(G)$.

Define:

\begin{center}$H_{1}\equiv (H\backslash\lbrace (u,u_{1}),(v,v_{1})\rbrace )\cup 
\lbrace e,(w,u_{1})\rbrace : $\end{center}

Note that:

\begin{center}$H_{1}\cap H^{'}=\emptyset , $\end{center}

\begin{center}$\vert H_{1}\vert +\vert H^{'}\vert =\vert H\vert +\vert H^{'}\vert 
=\nu _{2}(G), $\end{center}

\begin{center}$\vert H_{1}\vert =\vert H\vert =\alpha _{2}(G)$,
\end{center}

but

\begin{center}$\vert M\cap (H_{1}\cup H^{'})\vert >\vert M\cap (H\cup H^{'})\vert ,$\end{center}

which contradicts the choice of the pair $((H,H^{'}),M)$.

Sub-case B1.2: If $(w,u_{2})\in E(G)$.

Define:

\begin{center}$H_{1}\equiv (H\backslash(E(P_{w,u})\cup \lbrace (v,v_{1})\rbrace 
))\cup (H^{'}\cap E(P_{w,u}))\cup \lbrace e\rbrace ,$\end{center}

\begin{center}$H_{1}^{'}\equiv (H^{'}\backslash(E(P_{w,u})\cup \lbrace 
(u,u_{2})\rbrace ))\cup (H\cap E(P_{w,u}))\cup \lbrace (w,u_{2})\rbrace 
,$\end{center}

where recall that $P_{w,u}$ is the subpath of $P$, which connects the vertices 
$w$ and $u$.

Note that:

\begin{center}$H_{1}\cap H_{1}^{'}=\emptyset , $\end{center}

\begin{center}$\vert H_{1}\vert +\vert H_{1}^{'}\vert =\vert H\vert +\vert H^{'}\vert 
=\nu _{2}(G), $\end{center}

\begin{center}$\vert H_{1}\vert =\vert H\vert =\alpha _{2}(G)$,\end{center}

but:

\begin{center}$\vert M\cap (H_{1}\cup H_{1}^{'})\vert >\vert M\cap (H\cup H^{'})\vert 
,$\end{center}

which contradicts the choice of the pair $((H,H^{'}),M)$.

Sub-case B2: The vertex $w$ is incident to an edge from the set $H\cap M$.

As $P\in P_{e}(H,H^{'})$, then by exchanging the edges of the path $P$ 
between sets $H$ and $H^{'}$, we get a new pair $(H_{1},H_{1}^{'})$ of edge-disjoint pairs of matchings, that satisfies to conditions

\begin{center}$(H_{1},H_{1}^{'})\in M_{2}(G),$\end{center}

\begin{center}$\vert M\cap (H_{1}\cup H_{1}^{'})\vert =\vert M\cap 
(H\cup H^{'})\vert $ is maximized.\end{center}

Moreover, in this case the vertex $w$ is incident to an edge from $
H_{1}^{'}\cap M$, and therefore it brings to the sub-case B1.
\end{proof}

\begin{corollary}\label{MAH'PathsCaseCorollary} If $P\in Y(M,H,H^{'})$, then for every 
$v\in V_{0}(P)$ and $w\in V(P^{'})$, where $P'\in P_{0}^{H^{'}}(H^{'}, M_{A})$, if $
(v,w)\in E(G)$, then $\vert P'\vert =1$.
\end{corollary}

Let $P_{Y}\in Y(M,H,H')$ and $P\in P_{o}^{M}(M,H)$ be two paths such that $
V_{0}(P_{Y})\cap V_{0}(P)\ne \emptyset $. Note that we have only one 
vertex with $v\in V_{0}(P_{Y})\cap V_{0}(P)$; moreover, $v$ is 
incident to an edge from $H'$. Denote by $w_{0}$ the vertex, 
for which we have $w_{0}\in V(P_{Y})\cap V(P)$, $w_{0}$ is incident to an
edge from $M_{A}\backslash H'$ and all vertices of the path $P_{v,w_{0}}$
, which connects the vertices $v$ and $w_{0}$, are from $V(P_{Y})\cap V(P)$.

\begin{lemma}\label{TwoYEndsEdge} Let $v_{1}\in V_{0}(P_{1})$ and $
v_{2}\in V_{0}(P_{2})$, where $P_{1}, P_{2}\in Y(M,H,H')$. Then, there is no edge $(v_{1},v_{2})\in E(G)$.
\end{lemma}

\begin{proof} Assume the contrary. Note that if $P_{1}\ne P_{2}$, 
then we can exchange the edges of $H$ and $H'$ on the paths $P_{1}$ and $P_{2}$, so that the  edges incident to $v_1$ and $v_2$ become from $H'$. 

Having done this, we can add $(v_{1},v_{2})$ to $H$ and increase the number of edges in the $H\cup H'$, which will contradict the choice of the pair $(H,H')$. 

So, we can assume that $P_{1}=P_{2}$ (figure \ref{TwoYEndsEdgeFigure}).

\begin{center}
\begin{figure}[h]
\centering
\includegraphics[width=14.69cm,height=8.14cm]{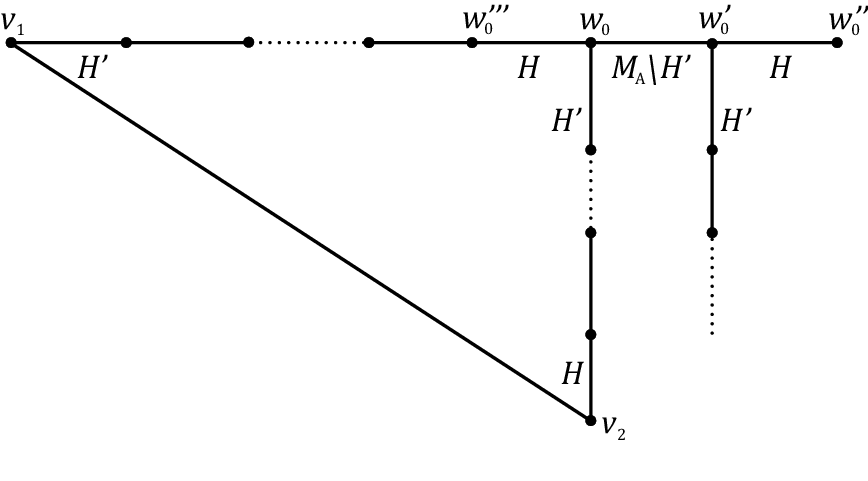}
\caption{There can be no edge $(v_1,v_2)$}
\label{TwoYEndsEdgeFigure}
\end{figure}
\end{center}

Define:

\begin{center}$H_{1}\equiv (H\backslash(E(P_{v_{1}w_{0}})\cup \lbrace ( 
w_{0}^{'},w_{0}^{''})\rbrace ))\cup (H^{'}\cap E(P_{v_{1}w_{0}}))\cup 
\lbrace (w_{0},w_{0}^{'})\rbrace ,$\end{center}

\begin{center}$H_{1}^{'}\equiv (H^{'}\backslash E(P_{v_{1}w_{0}}))\cup (H\cap 
E(P_{v_{1}w_{0}^{'''}}))\cup \lbrace (v_{1},v_{2})\rbrace :$\end{center}

Note, that:

\begin{center}$H_{1}\cap H_{1}^{'}=\emptyset ,$\end{center}

\begin{center}$\vert H_{1}\vert +\vert H_{1}^{'}\vert =\vert H\vert 
+\vert H^{'}\vert =\nu _{2}(G),$ \end{center}

\begin{center}$|H_{1}|=\alpha _{2}(G),$ \end{center}

but

\begin{center} $\vert M\cap (H_{1}\cup H_{1}^{'})\vert >\vert M\cap 
(H\cup H^{'})\vert $,\end{center}

which contradicts the choice of the pair $((H,H^{'}),M)$. Thus $(v_{1},v_{2})\notin E(G)$.
\end{proof}

\begin{lemma}\label{vwEdgeLyingOutside} Let $v\in V_{0}(P)$, where $P\in Y(M,H,H')$ and $(v,w)$ is an edge with $(v,w)\notin E(P)$. Then, $w$ is incident to edges from $H$ and $H'$, which we will denote by $e_{w}(H)$ and $e_{w}(H^{'})$, respectively. Moreover, one of these edges is from $M$.
\end{lemma}

\begin{proof} Assume that $v$ is incident to an edge from $H^{'}$ (the other case can be solved in a similar way). First of all, let us show that $w$ is incident to an edge from $H$. If $w$ is not incident to an edge from $H$, then we can add the edge $(v,w)$ to $H$ 
and increase the number of edges in $H\cup H'$, which would contradict the choice of the pair $(H,H')$. Thus, $w$ is incident to an edge from $H$. 

Now, we will show that $w$ is incident to an edge from $H'$. If $w$ were not incident to an edge from $H'$, then as follows from lemma \ref{TwoYEndsEdge}, $w$ is not the 
other end-vertex of $P$, thus we can exchange the edges of $P$ between $H$ 
and $H'$, and as a result, we can add the edge $(v,w)$ to $H'$ and 
increase the number of edges in $H\cup H'$, which would contradict the choice of the pair $(H,H')$. Thus, $w$ is incident to an edge from $H'$, too.

Finally, we will show that one of edges $e_{w}(H)$ and $e_{w}(H^{'})$ belongs to $M$. If neither of edges $e_{w}(H)$ and $e_{w}(H^{'})$ were from $M$, then we would have a path from $P_{e}(M,H)$ or $P_{o}^{H}(M,H)$ beginning from the vertex $w$, which would contradict lemma \ref{MHaltpaths}. Thus, one of edges $e_{w}(H)$ and $e_{w}(H^{'})$ is from $M$.
\end{proof}

\begin{lemma}\label{SecondPath} Let $P\in P_{o}^{M}(M,H)$, $P_{Y}\in Y(M,H,H')$ and $V_0(P)\cap V_0(P_{Y})\neq \emptyset$. Then, if $P_{Y}$ contains edges from another path $P'\in P_{o}^{M}(M,H)$ ($P'\neq P$), then $P'$ has at least two edges from $M_{A}\backslash H^{'}$. 
\end{lemma}

\begin{proof} As we have noted in the proof of lemma \ref{MHpaths}, each path from $P_{o}^{M}(M,H)$ contains an edge from $M_{A}\backslash H^{'}$. Now, if we assume that $P'$ has exactly one edge from $M_{A}\backslash H^{'}$, then $P_Y$ must intersect at least one of two paths from $Y(M,H,H')$ that begins from a vertex of $V_0(P')$, which would contradict the fact that the paths of $Y(M,H,H')$ are disjoint (proof of lemma \ref{H'HpathsandPoddMAH'}(a)).
\end{proof}

\begin{theorem}\label{MainBound} For every cubic graph $G$ the 
inequality $\frac{\nu (G)}{\alpha _{2}(G)}\le \frac{9}{8}$ holds.
\end{theorem}

\begin{proof} From lemma \ref{equality}, we have:

\begin{center}$\vert H^{'}\vert =\vert P_{o}^{H^{'}}(M_{A},H^{'})\vert 
+\vert H_{A}\vert +\nu (G)-\alpha _{2}(G)$.\end{center}

We claim that:
\begin{equation}\label{inequality1}
\vert P_{o}^{H^{'}}(M_{A},H^{'})\vert +\vert H_{A}\vert 
\ge 7(\nu (G)-\alpha _{2}(G)).
\end{equation}

As $\vert H^{'}\vert \le \vert H\vert =\alpha _{2}(G)$, then from inequality (\ref{inequality1}) we have:

\begin{center}$\alpha _{2}(G)\ge \vert H^{'}\vert \ge 8(\nu 
(G)-\alpha _{2}(G))$,\end{center}

which is the same, as:

\begin{center}$\frac{\nu (G)}{\alpha _{2}(G)}\le \frac{9}{8}$.
\end{center}

Define:

\begin{center}$V_{Y}\equiv \bigcup_{P\in Y(M,H,H^{'})}{V_{0}(P)}$\end{center}

Note that by lemma \ref{H'HpathsandPoddMAH'}(a):

\begin{center}$\vert V_{Y}\vert =2\vert Y\vert =4(\nu (G)-\alpha 
_{2}(G))$.\end{center}

Therefore for proving (\ref{inequality1}), we only need to show, that:
\begin{equation}\label{inequality2}
\vert P_{o}^{H^{'}}(M_{A},H^{'})\vert +\vert H_{A}\vert 
\ge 3(\nu (G)-\alpha _{2}(G))+\vert V_{Y}\vert.
\end{equation}

For proving (\ref{inequality2}), let us show, that we can partition the set $P_{o}^{H^{'}}(M_{A},H^{'}) \cup H_{A}$ into two sets $S_{1}$ and $S_{2}$, such that $S_{1}\cap S_{2}=\emptyset $, $S_{1}\cup S_{2}=P_{o}^{H^{'}}(M_{A},H^{'})\cup H_{A}$, $\vert 
S_{1}\vert \ge 3(\nu (G)-\alpha _{2}(G))$ and $\vert S_{2}\vert 
=\vert V_{Y}\vert $.

First of all, we will give the description of the set $S_{2}$. For every vertex $
v\in V_{Y}$ conform one edge from $H_{A}$ or one path from $P_{o}^{H^{'}}(M_{A},H^{'})$, such that for every pair of vertices $v_{1},v_{2}\in V_{Y}$, which are different from each other, we have 
different edges or paths from $ P_{o}^{H^{'}}(M_{A},H^{'})\cup  H_{A}$.

Consider a vertex $v\in V_{Y}$ and an edge $(v,w)$, with $(v,w)\notin E(P)$, 
where $v\in V_{0}(P)$, $P\in Y$. From lemma \ref{vwEdgeLyingOutside}, we 
have, that the vertex $w$ is incident to edges from $H$ and $H'
$, which we have denoted by $e_{w}(H)$ and $e_{w}(H^{'})$,
respectively.

Consider a mapping $f:V_{Y}\to P_{o}^{H^{'}}(M_{A},H^{'})\cup H_{A}$, 
where for every $v\in V_{Y}$:\\

$f(v)=\begin{cases}\mbox{the path from } P_{o}^{H^{'}}(M_{A},H^{'})\mbox{ which contains the edge }
e_{w}(H^{'}), & \mbox{if } e_{w}(H)\in M \\ e_{w}(H), & \mbox{if } e_{w}(H^{'})\in M\end{cases}$\\

Denote by $S_{2}$ the image set of the mapping $f$. Note, 
that the lengths of paths from $S_{2}$ are one, which follows from 
corollary \ref{MAH'PathsCaseCorollary}. Thus, $S_{2}\subseteq H\cup H'$.

We claim that the mapping $f:V_{Y}\to P_{o}^{H^{'}}(M_{A},H^{'})\cup H_{A}$ is injective. Suppose on the contrary that there are two vertices $
v_{1},v_{2}\in V_{Y},(v_{1}\ne v_{2})$, for which $f(v_{1})=f(v_{2})$. Note that since the graph $G$ is cubic, this can occur only in that case, when for some $w_{1}$ and $w_{2}$ ($w_{1}\ne w_{2}$), we have $(v_{1},w_{1})\in E(G)$, $(v_{2},w_{2})\in E(G)$ and $(w_{1},w_{2})\in E(G)$.

Lemma \ref{OddPathCase} implies that $(w_{1},w_{2})\notin E(P^{'})
$ for every $P'\in P_{o}^{H}(H,H')$. Lemma \ref{CycleCaseTwoPathsY} implies that $(w_{1},w_{2})\notin E(C^{'})$ for every $C'\in C(H,H')$.

Lemma \ref{EvenPathCase} implies that $(w_{1},w_{2})\notin E(P^{'})
$ for every $P'\in P_{e}(H,H')$ with $P^{'}\ne P_{1}$ and $P^{'}\ne P_{2}$, where $v_{1}\in V_{0}(P_{1})$, $v_{2}\in V_{0}(P_{2})$ and $P_1, P_2 \in Y$. Thus, it remains to consider the case when the edge $(w_{1},w_{2})$ lies on at least one of paths $P_{1}$ and $P_{2}$. Without loss of generality, we can assume that $(w_{1},w_{2})$ lies on $P_{1}$.

We will consider four cases:\\

Case A: The vertex $v_{1}$ is incident to an edge from $H'$ and the edge $
(w_{1},w_{2})$ belongs to the subpath $P_{v_{1},w_{0}}$ of $P_{1}$. 

Note that, if $(w_{1},w_{2})\in H'$, then from definition of $f$ we 
have, that $f(v_{1})\ne f(v_{2})$, because every edge of $
P_{v_{1},w_{0}}$, which belongs to $H'$, also belongs to $M$.

Therefore, for this case we will consider only the situation, when $
(w_{1},w_{2})\in H$. On the other hand, note that lemma \ref{EvenPathCase} 
implies, that distance of $w_{2}$ from end-vertices of $P_{1}$ is 
odd (figure \ref{CaseAFigure}).

\begin{center}
\begin{figure}[h]
\centering
\includegraphics[width=17.09cm,height=4.81cm]{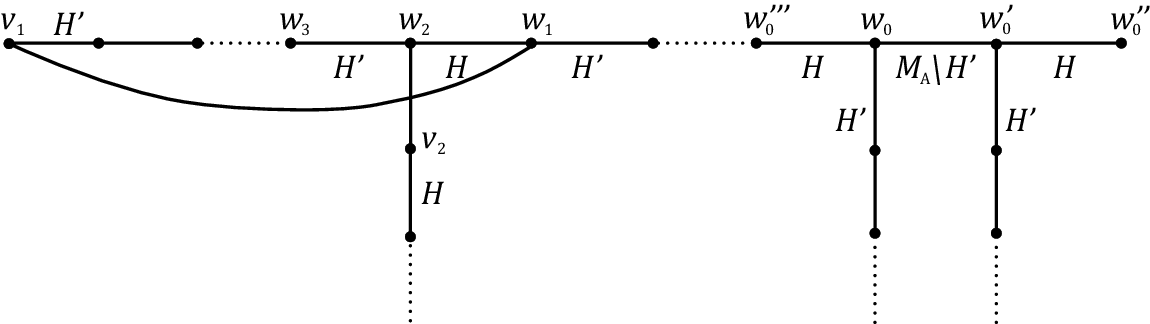}
\caption{Theorem \ref{MainBound}, Case A}
\label{CaseAFigure}
\end{figure}
\end{center}

Define:

\begin{center}$H_{1}\equiv (H\backslash(E(P_{v_{1}w_{0}})\cup \lbrace 
( w_{0}^{'},w_{0}^{''})\rbrace ))\cup (H^{'}\cap E(P_{v_{1}w_{0}}))\cup 
\lbrace (w_{0},w_{0}^{'})\rbrace $,\end{center}

\begin{center}$H_{1}^{'}\equiv (H^{'}\backslash E(P_{v_{1}w_{0}}))\cup 
(H\cap (E(P_{v_{1}w_{0}^{'''}})\backslash\lbrace (w_{1},w_{2})\rbrace 
))\cup \lbrace (v_{1},w_{1}),(v_{2},w_{2})\rbrace $.\end{center}

Note that:

\begin{center}$H_{1}\cap H_{1}^{'}=\emptyset ,$\\
\end{center}

\begin{center}$\vert H_{1}\vert +\vert H_{1}^{'}\vert =\vert H\vert 
+\vert H^{'}\vert =\nu _{2}(G),$ \end{center}

\begin{center}$\vert H_{1}\vert=\alpha _{2}(G),$\end{center}

but:

\begin{center}$\vert M\cap (H_{1}\cup H_{1}^{'})\vert >\vert M\cap 
(H\cup H^{'})\vert $,\end{center}

which contradicts the choice of the pair $((H,H^{'}),M)$.

Case B: The vertex $v_{1}$ is incident to an edge from $H$ and the edge $
(w_{1},w_{2})$ belongs to $P_{1}$, but not to the subpath $
P_{v_{1},w_{0}}$ of $P_{1}$. 

Denote by $u_{1}$ the other end-vertex of $P_{1}$. Note that, if $(w_{1},w_{2})\in H'$, then from definition of $f$ we will have, that $f(v_{1})\ne f(v_{2})$, because every edge 
of $P_{u_{1},w_{0}}$, which belongs to $H'$, also belongs to $M$.

Therefore, for this case we will consider only the situation, when $
(w_{1},w_{2})\in H$. On the other hand, note that lemma \ref{EvenPathCase} implies 
that the vertex $w_{1}$ belongs to the subpath of $P_{1}$, 
which connects the vertices $w_{2},w_{0},v_{1}$ in this order (figure \ref{CaseBFigure}).

\begin{center}
\begin{figure}[h]
\centering
\includegraphics[width=17.09cm,height=10.08cm]{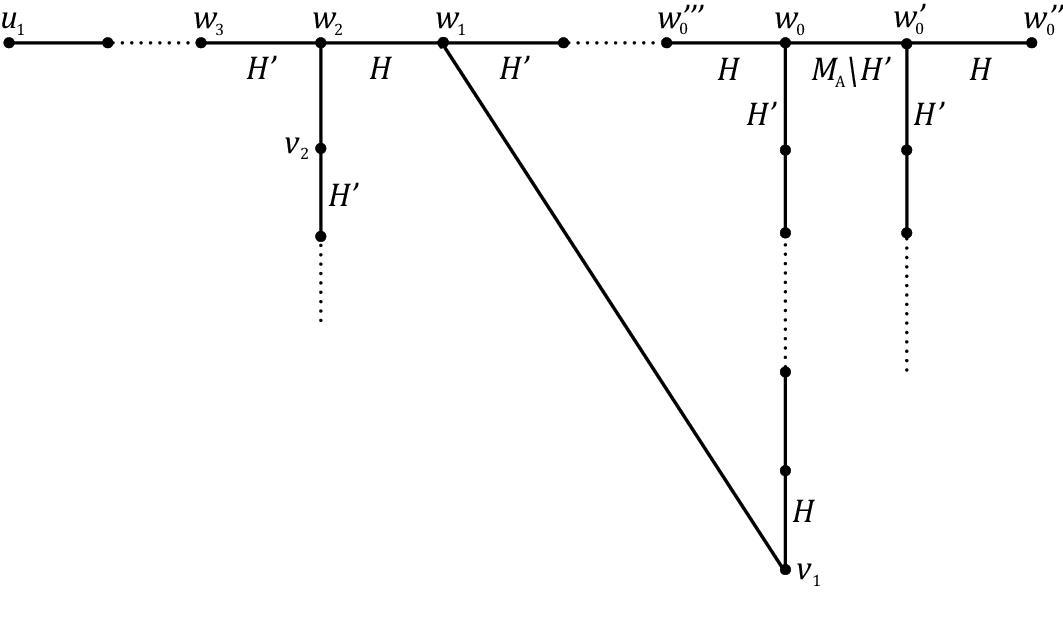}
\caption{Theorem \ref{MainBound}, Case B}
\label{CaseBFigure}
\end{figure}
\end{center}

Define:

\begin{center}$H_{1}\equiv (H\backslash(E(P_{w_{2}w_{0}})\cup \lbrace 
( w_{0}^{'},w_{0}^{''})\rbrace ))\cup (H^{'}\cap E(P_{w_{2}w_{0}}))\cup 
\lbrace (w_{0},w_{0}^{'}),(v_{2},w_{2})\rbrace $,\end{center}

\begin{center}$H_{1}^{'}\equiv (H^{'}\backslash E(P_{w_{2}w_{0}}))\cup 
(H\cap E(P_{w_{1}w_{0}^{'''}}))\cup \lbrace (v_{1},w_{1})\rbrace $.
\end{center}

Note that:

\begin{center}$H_{1}\cap H_{1}^{'}=\emptyset ,$\\
\end{center}

\begin{center}$\vert H_{1}\vert +\vert H_{1}^{'}\vert =\vert H\vert 
+\vert H^{'}\vert =\nu _{2}(G),$ \end{center}

\begin{center}$\vert H_{1}\vert=\alpha _{2}(G),$\end{center}

but:

\begin{center}$\vert M\cap (H_{1}\cup H_{1}^{'})\vert >\vert M\cap 
(H\cup H^{'})\vert $,\end{center}

which contradicts the choice of the pair $((H,H^{'}),M)$.

Case C: The vertex $v_{1}$ is incident to an edge from $H'$ and the edge $
(w_{1},w_{2})$ belongs to $P_{1}$, but not to the subpath $
P_{v_{1},w_{0}}$ of $P_{1}$.

We will consider two sub-cases:

Sub-case C.1: $(w_{1},w_{2})\in H^{'}$. 

In this case Lemma \ref{EvenPathCase} implies that the vertex $w_{1}$ 
belongs to the subpath of $P_{1}$, which connects the vertices $
v_{1}, w_{2}$. Note that, if $(w_{1},w_{2})\in H'\cap M$, then from 
definition of $f$ we have, that $f(v_{1})\ne f(v_{2})$.

Therefore, for this case we will consider only the situation, when $
(w_{1},w_{2})\notin M$. We can suppose, that the vertex $v_{2}$ is 
incident to an edge from $H$ (figure \ref{CaseC1Figure}).

\begin{center}
\begin{figure}[h]
\centering
\includegraphics[width=14.34cm,height=12.51cm]{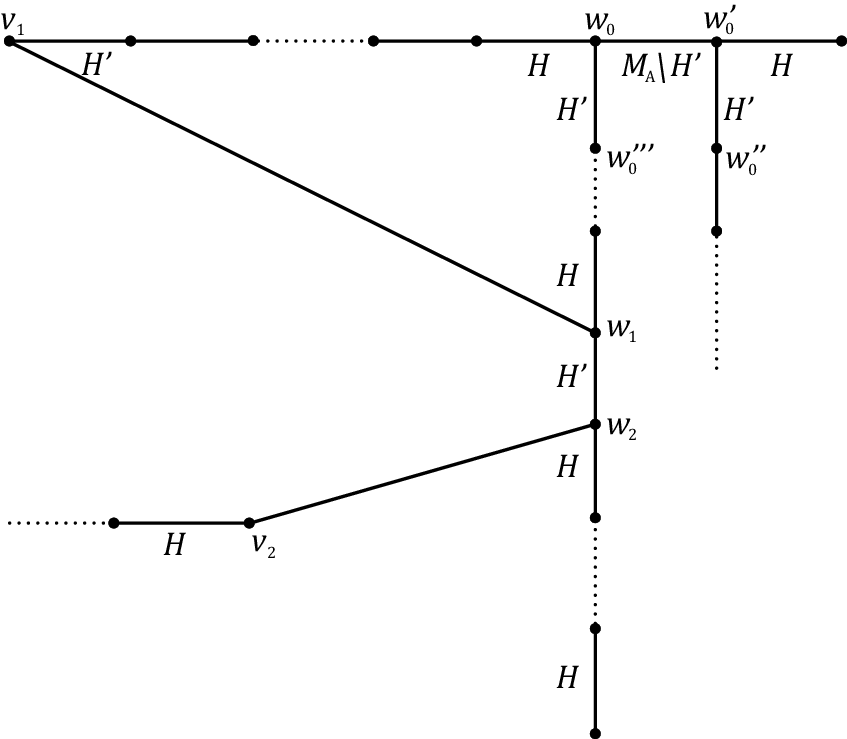}
\caption{Theorem \ref{MainBound}, sub-case C1}
\label{CaseC1Figure}
\end{figure}
\end{center}

Define:

\begin{center}$H_{1}\equiv (H\backslash E(P_{w_{0}w_{1}}))\cup 
(H^{'}\cap E(P_{w_{0}^{'''}w_{1}}))\cup \lbrace (v_{1},w_{1})\rbrace $,
\end{center}

\begin{center}$H_{1}^{'}\equiv (H^{'}\backslash(E(P_{w_{0}w_{2}})\cup 
\lbrace ( w_{0}^{'},w_{0}^{''})\rbrace ))\cup (H\cap 
E(P_{w_{0}w_{1}}))\cup \lbrace (w_{0},w_{0}^{'}),(v_{2},w_{2})\rbrace $
.\end{center}

Note that:

\begin{center}$H_{1}\cap H_{1}^{'}=\emptyset ,$\\
\end{center}

\begin{center}$\vert H_{1}\vert +\vert H_{1}^{'}\vert =\vert H\vert 
+\vert H^{'}\vert =\nu _{2}(G),$ \end{center}

\begin{center}$\vert H_{1}\vert=\alpha _{2}(G),$\end{center}

but:

\begin{center}$\vert M\cap (H_{1}\cup H_{1}^{'})\vert >\vert M\cap 
(H\cup H^{'})\vert $,\end{center}

which contradicts the choice of the pair $((H,H^{'}),M)$.

Sub-case C.2: $(w_{1},w_{2})\in H$. 

In this case lemma \ref{EvenPathCase} implies that the vertex $w_{2}$ belongs 
to the subpath of $P_{1}$, which connects the vertices $v_{1}, 
w_{1}$. Note that, if $(w_{1},w_{2})\in H\cap M$, then from definition 
of $f$ we will have, that $f(v_{1})\ne f(v_{2})$.

Therefore, for this case we will consider only the situation, when $
(w_{1},w_{2})\notin M$. We can suppose, that the vertex $v_{2}$ is 
incident to an edge from $H$ (figure \ref{CaseC2Figure}).

\begin{center}
\begin{figure}[h]
\centering
\includegraphics[width=14.34cm,height=12.51cm]{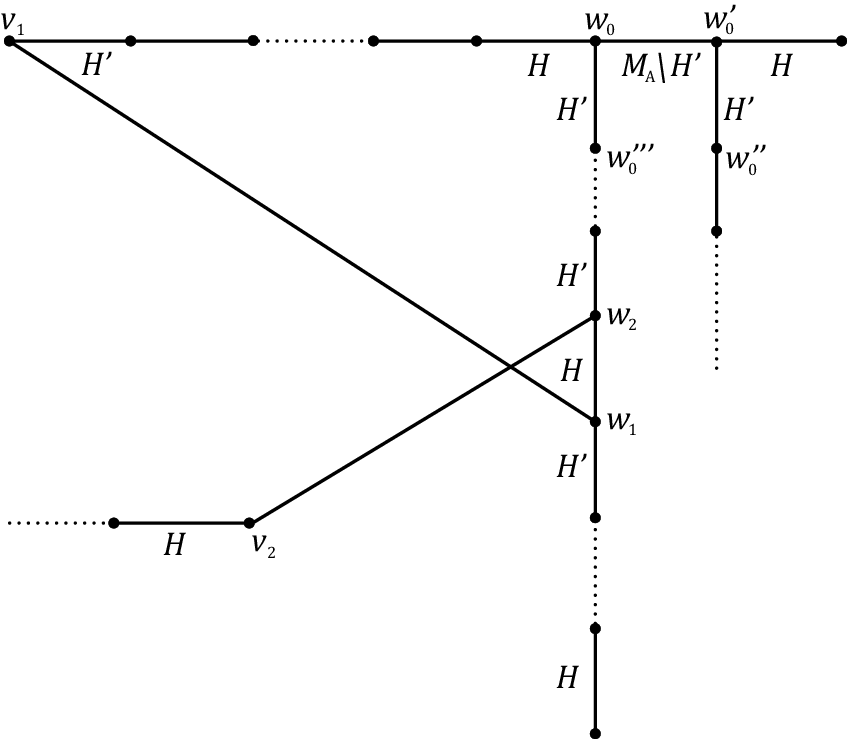}
\caption{Theorem \ref{MainBound}, sub-case C2}
\label{CaseC2Figure}
\end{figure}
\end{center}

Define:

\begin{center}$H_{1}\equiv (H\backslash E(P_{w_{0}w_{1}}))\cup 
(H^{'}\cap E(P_{w_{0}^{'''}w_{2}}))\cup \lbrace (v_{1},w_{1})\rbrace $,
\end{center}

\begin{center}$H_{1}^{'}\equiv (H^{'}\backslash(E(P_{w_{0}w_{2}})\cup 
\lbrace ( w_{0}^{'},w_{0}^{''})\rbrace ))\cup (H\cap 
E(P_{w_{0}w_{2}}))\cup \lbrace (w_{0},w_{0}^{'}),(v_{2},w_{2})\rbrace $
.\end{center}

Note that:

\begin{center}$H_{1}\cap H_{1}^{'}=\emptyset ,$\end{center}

\begin{center}$\vert H_{1}\vert +\vert H_{1}^{'}\vert =\vert H\vert 
+\vert H^{'}\vert =\nu _{2}(G),$ \end{center}

\begin{center}$\vert H_{1}\vert=\alpha _{2}(G),$\end{center}

but:

\begin{center}$\vert M\cap (H_{1}\cup H_{1}^{'})\vert >\vert M\cap 
(H\cup H^{'})\vert $,\end{center}

which contradicts the choice of the pair $((H,H^{'}),M)$.

Case D: The vertex $v_{1}$ is incident to an edge from $H$ and the edge $
(w_{1},w_{2})$ belongs to the subpath $P_{v_{1},w_{0}}$ of $P_{1}$.

We will consider two sub-cases:

Sub-case D.1: $(w_{1},w_{2})\in H^{'}$. 

In this case lemma \ref{EvenPathCase} implies that the vertex $w_{2}$ 
belongs to the subpath of $P_{1}$, which connects the vertices $
v_{1}, w_{1}$. Note that, if $(w_{1},w_{2})\in H'\cap M$, then from 
definition of $f$ we will have, that $f(v_{1})\ne f(v_{2})$.

Therefore, for this case we will consider only the situation, when $
(w_{1},w_{2})\notin M$. We can suppose, that the vertex $v_{2}$ is 
incident to an edge from $H'$ (figure \ref{CaseD1Figure}).

\begin{center}
\begin{figure}[h]
\centering
\includegraphics[width=14.27cm,height=13.63cm]{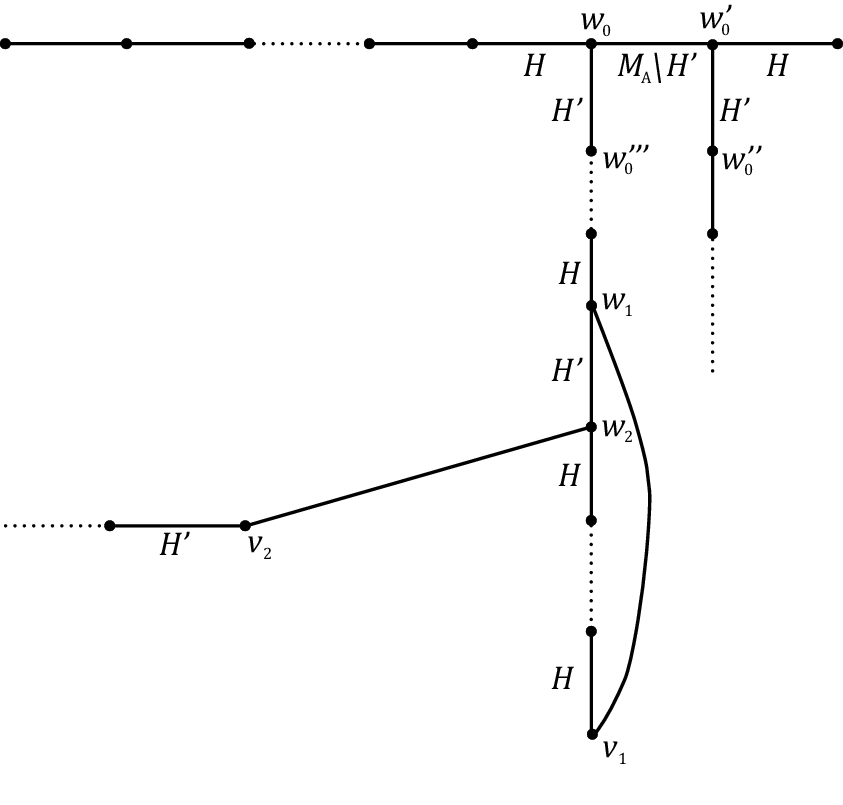}
\caption{Theorem \ref{MainBound}, sub-case D1}
\label{CaseD1Figure}
\end{figure}
\end{center}

Define:

\begin{center}$H_{1}\equiv (H\backslash E(P_{w_{0}v_{1}}))\cup 
(H^{'}\cap (E(P_{w_{0}^{'''}w_{1}})\cup E(P_{v_{1}w_{2}})))\cup \lbrace 
(v_{1},w_{1}),(v_{2},w_{2})\rbrace $,\end{center}

\begin{center}$H_{1}^{'}\equiv (H^{'}\backslash(E(P_{w_{0}v_{1}})\cup 
\lbrace ( w_{0}^{'},w_{0}^{''})\rbrace ))\cup (H\cap 
E(P_{w_{0}v_{1}}))\cup \lbrace (w_{0},w_{0}^{'})\rbrace $.\end{center}

Note that:

\begin{center}$H_{1}\cap H_{1}^{'}=\emptyset ,$\\
\end{center}

\begin{center}$\vert H_{1}\vert +\vert H_{1}^{'}\vert =\vert H\vert 
+\vert H^{'}\vert =\nu _{2}(G),$ \end{center}

\begin{center}$\vert H_{1}\vert=\alpha _{2}(G),$\end{center}

but:

\begin{center}$\vert M\cap (H_{1}\cup H_{1}^{'})\vert >\vert M\cap 
(H\cup H^{'})\vert $,\end{center}

which contradicts the choice of the pair $((H,H^{'}),M)$.

Sub-case D.2: $(w_{1},w_{2})\in H$. 

In this case lemma \ref{EvenPathCase} implies that the vertex $w_{1}$ belongs 
to the subpath of $P_{1}$, which connects the vertices $v_{1}, w_{2}
$. Note that, if $(w_{1},w_{2})\in H\cap M$, then from definition of $
f$ we will have, that $f(v_{1})\ne f(v_{2})$.

Thus, for this case we will consider only the situation, when $
(w_{1},w_{2})\notin M$. We can suppose, that the vertex $v_{2}$ is 
incident to an edge from $H'$. 

As $(w_{1},w_{2})\in H$ and $(w_{1},w_{2})\notin M$, then there is $P'\in P_{o}^{M}(M,H)$, such that $(w_{1},w_{2})\in E(P')$. On the other hand, the end-edge of $P_{1}$
, which is incident to the vertex $v_{1}$, belongs to $H\cap M$, 
therefore the path $P'$ cannot contain the overall subpath $
P_{v_{1},w_{2}}$. As a result, there is $w'\in P_{v_{1},w_{2}}$, for 
which $e_{w'}(H)\in M$ and $e_{w'}(H^{'})\in M$, therefore for the third 
edge $(w^{'},w_{1}^{'})$, which is incident to $w^{'}$, $
(w^{'},w_{1}^{'})\in E(P')$; moreover $
(w^{'},w_{1}^{'})\in M_{A}\backslash H'$ (figure \ref{CaseD2Figure}). 

\begin{center}
\begin{figure}[h]
\centering
\includegraphics[width=12.27cm,height=12.89cm]{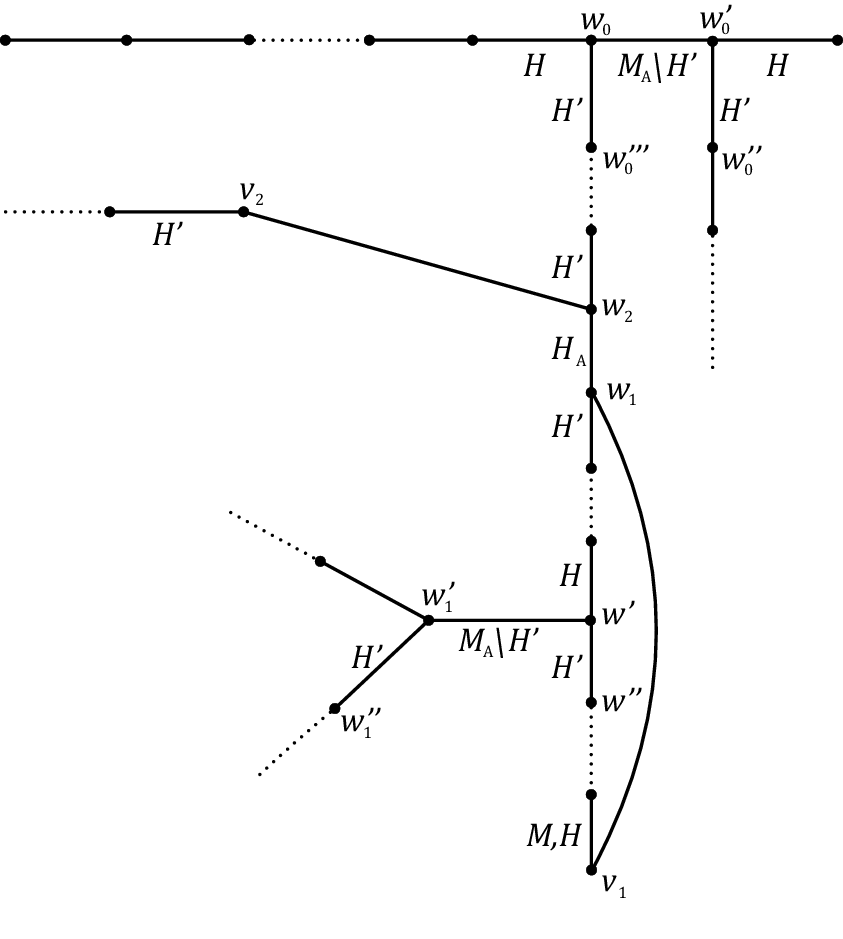}
\caption{Theorem \ref{MainBound}, sub-case D2}
\label{CaseD2Figure}
\end{figure}
\end{center}

Define:

\begin{center}$H_{1}\equiv (H\backslash(E(P_{w^{'}v_{1}})\cup \lbrace 
(w_{1},w_{2})\rbrace ))\cup (H^{'}\cap E(P_{w^{''}v_{1}}))\cup \lbrace 
(v_{1},w_{1}),(v_{2},w_{2})\rbrace $,\end{center}

\begin{center}$H_{1}^{'}\equiv (H^{'}\backslash(E(P_{w^{'}v_{1}})\cup 
\lbrace ( w_{1}^{'},w_{1}^{''})\rbrace ))\cup (H\cap 
E(P_{w^{'}v_{1}}))\cup \lbrace (w^{'},w_{1}^{'})\rbrace $.\end{center}

Note that:

\begin{center}$H_{1}\cap H_{1}^{'}=\emptyset ,$\\
\end{center}

\begin{center}$\vert H_{1}\vert +\vert H_{1}^{'}\vert =\vert H\vert 
+\vert H^{'}\vert =\nu _{2}(G),$ \end{center}

\begin{center}$\vert H_{1}\vert=\alpha _{2}(G),$\end{center}

but:

\begin{center}$\vert M\cap (H_{1}\cup H_{1}^{'})\vert >\vert M\cap 
(H\cup H^{'})\vert $,\end{center}

which contradicts the choice of the pair $((H,H^{'}),M)$.

Thus, the mapping $f$ is injective. \\

Finally, let us denote by $S_{1}=(P_{o}^{H^{'}}(M_{A},H^{'})\cup 
H_{A})\backslash S_{2}$ and show that $\vert S_{1}\vert \ge 3(\nu 
(G)-\alpha _{2}(G))$.

Due to lemma \ref{MANeighbourhood} and the definition of function $f$, 
the edges from $H_{A}$, which belong to $S_{2}$ are not adjacent 
to the edges of $M_{A}\backslash H'$.

On the other hand, all paths from $P_{o}^{M}(M,H)$ can be divided into the following
two sets:

\newcounter{numberedCntD}
\begin{enumerate}
\item Paths which contain only one edge from $M_{A}\backslash H'$.
\item Paths which contain at least two edges from $M_{A}\backslash H'$
.
\setcounter{numberedCntD}{\theenumi}
\end{enumerate}
The paths of type (a) have two edges from $
H_{A}$ and one path $P'$ from $P_{o}^{H^{'}}(M_{A},H^{'})$. The 
length of $P'$ is at least three, therefore it is not from $S_{2}$. If 
there is another path from $P_{o}^{M}(M,H)$, which contains edges 
from $P'$, then it will contain edges of path from $Y(M,H,H')$, 
one of end-edges which and our path are the same, therefore from 
lemma \ref{SecondPath}, we will have, that the second path from $
P_{o}^{M}(M,H)$ belongs to the set (b).

As the paths with type (b) contain at least two 
edges from $M_{A}\backslash H'$, therefore they have at least three 
edges from $H_{A}$.

As $\vert P_{o}^{M}(M,H)\vert =\nu (G)-\alpha 
_{2}(G)$, we get $\vert S_{1}\vert \ge 3(\nu (G)-\alpha _{2}(G))
$. The proof of theorem \ref{MainBound} is completed.
\end{proof}

\begin{remark}
Note that the bound of theorem \ref{MainBound} cannot be improved. This follows from the
example from figure \ref{upper bound example}.
\end{remark}

\begin{center}
\begin{figure}[h]
\centering
\includegraphics[width=17.09cm,height=5.12cm]{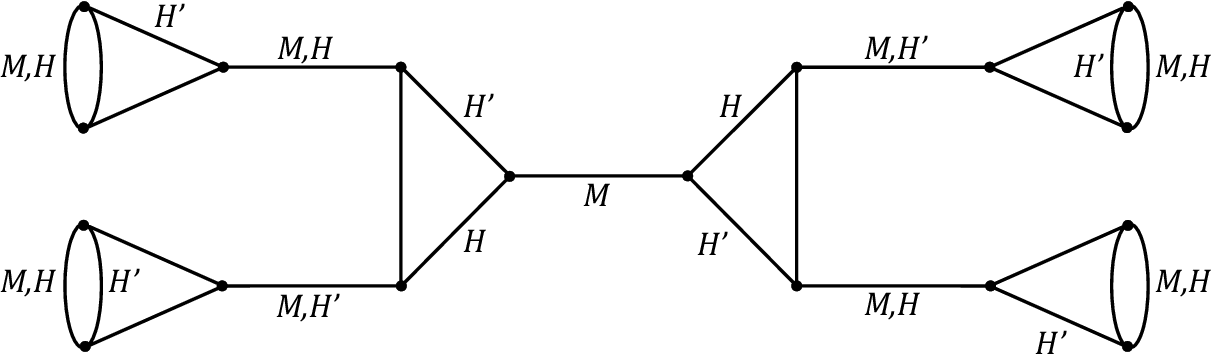}
\caption{A cubic graph $G$ with $\nu (G)=9, \nu _{2}(G)=16,\alpha _{2}(G)=8$}
\label{upper bound example}
\end{figure}
\end{center}

\begin{acknowledgement}We thank our reviewers for their comments that helped us to improve the paper, and particularly for shortening the proof of theorem \ref{nu2nu3}. In the paper we have reproduced one of these proofs. And also, our special thanks to the reviewer who has pointed out a mistake in the proof of theorem \ref{2factorExtension}.
\end{acknowledgement}

\end{document}